\documentclass[journal,onecolumn]{IEEEtran}

\ifCLASSINFOpdf
  
\else
  
\fi

\usepackage{amsmath}
\usepackage{amsthm}
\usepackage{cases}
\usepackage{amsfonts}
\usepackage{cite}
\usepackage{amssymb}
\usepackage{enumerate}
\usepackage{graphicx}
\usepackage{float}
\usepackage{caption}
\usepackage{verbatim}
\usepackage{dblfloatfix}
\usepackage[cmintegrals]{newtxmath}
\usepackage[ruled,linesnumbered]{algorithm2e}
\usepackage{color}
\usepackage{soul}

\newtheorem{theorem}{Theorem}
\newtheorem{lemma}{Lemma}
\newtheorem{definition}{Definition}
\newtheorem{corollary}{Corollary}
\newtheorem{remark}{Remark}

\newtheorem{example}{Example}
\newtheorem{applemma}{Lemma}[section]
\begin{document}

\title{Robust Joint Message and State Communication Under Arbitrarily Varying Jamming}

\author{
  \IEEEauthorblockN{Yiqi Chen,
  Holger Boche}\\
  \IEEEauthorblockA{
Technical University of Munich,  
80333 Munich, Germany \\
\{yiqi.chen, boche\}@tum.de}
}

\maketitle
\thispagestyle{empty}
\pagestyle{empty}

\begin{abstract}
Joint message and state transmission under arbitrarily varying jamming is investigated in this paper. The problem is modeled as the transmission over a channel with random states with a fixed distribution and jamming that varies in an unknown manner. We consider cases in which the channel state is known at the encoder strictly causally and noncausally. For each case, both the average error criterion and the maximal error criterion of the message transmission are adopted. The main results of this paper are lower bounds of the capacity--distortion function of the aforementioned scenarios. Some capacity-achieving cases are also provided. The proposed coding schemes are deterministic, and no correlated randomness is needed to achieve reliable communication and estimation. It turns out that the performance of the system under the average error can strictly outperform the maximal error case, which is in accordance with normal communication over arbitrarily varying channels.
\end{abstract}


\section{Introduction}
Integrated sensing and communication (ISAC) has shown great application prospects in the next generation of wireless communication networks, including drone communications, indoor localization, autonomous vehicle networks, and robotics \cite{fettweis20216g,liu2022survey,schwenteck20236g}. It combines two functionalities for future communication networks including information transmission via Shannon's channel model and environment detection via radar systems, and allows the two systems to share the same frequency band and hardware to improve the spectrum efficiency and reduce the energy cost\cite{liu2022survey}. Yet, a lot of research has been conducted from a signal processing perspective investigating the optimal beamforming design\cite{liu2021cramer}, the fundamental limits of the ISAC models in different scenarios remain open problems.

Research towards those unsolved problems for joint message communication and state sensing traces back to the joint message and state transmission \cite{sutivong2005channel,kim2008state,choudhuri2013causal}, which is also known as the \emph{state amplification} problem. The encoder observes the channel state in noncausal \cite{sutivong2005channel}\cite{kim2008state}, or causal \cite{choudhuri2013causal} manners and communicates to the receiver, aiming to reconstruct the state under a given distortion level. The cases in which the encoder has no access to the channel state are studied in \cite{zhang2011joint}, where the sensing receiver is integrated with the receiver, and in \cite{chen2025fundamental}, where the sensing receiver is separated from the receiver.

Another scenario, named mono-static ISAC, is the case in which the sensor is physically close to the encoder and hence can share all the information with the encoder. This line of work is initiated in \cite{kobayashi2018joint} for point-to-point cases and then extended to multi-user cases in \cite{ahmadipour2022information,li2024achievable,ahmadipour2023information}. Models beyond traditional communication and distortion-constrained sensing are also investigated, see \cite{gunlu2023secure,welling2024transmitter,chen2024distribution,wang2024covert} for ISAC with security constraint, \cite{joudeh2022joint,chang2023rate} for state detection, \cite{chen2025integrated} for ISAC with helpers and \cite{labidi2024joint}\cite{zhao2024identification} for joint identification and sensing.

Although significant effort has been invested in the research of ISAC, one of the most important properties of a communication system, robustness, has not yet been well studied. Robustness evaluates the performance of a communication system in the presence of jamming signals that are produced by a jammer who tries to corrupt the communication. Since the encoder and decoder are not aware of the jammer's jamming strategy, they have to make the worst assumption that the system can work reliably under any possible jamming strategy. Such a communication scenario is captured by the model arbitrarily varying channel (AVC), which was first studied in \cite{blackwell1959capacity}. The interference generated by the jammer is, of course, intentional and truly hostile. However, the AVC channel model also directly captures the influence of unintentional interference caused by a lack of coordination between different communication links. For these scenarios, communication should also work robustly independently of other uncoordinated communication links. This behavior is particularly required for application areas such as autonomous driving, drones, and autonomous robots. 

  The behavior of AVCs is significantly different from other stochastic channels with or without states in terms of capacity in the following two aspects: 1). The capacity of AVCs under the maximal error criterion can be smaller than that under the average error criterion, even for point-to-point channels\cite{ahlswede1980method}\cite{csiszar1981capacity}. 2). The capacity of AVCs without correlated randomness (CR) can be strictly smaller than the capacity with correlated randomness\cite{csiszar1988capacity}. These properties are caused by the zero capacity phenomenon of AVCs, which are the results of the symmetrizability under average error criterion\cite{ericson1985exponential}\cite{csiszar1988capacity} and complete graph under maximal error criterion\cite{ahlswede1980method}\cite{csiszar1981capacity}. In \cite{ahlswede1978elimination}, the capacities of AVCs with different available randomness are partially solved. It turns out that when there are no jamming constraints, AVCs without correlated randomness have either zero average error capacity or the same average error capacity as with correlated randomness.\cite{csiszar1988capacity} further extends the results by showing that in the presence of state constraints, the average error capacity of the AVC can be positive but strictly smaller than the correlated randomness-assisted capacity. The general capacity formula of AVCs under the maximal error criterion is still an open problem, and only the capacity of channels with binary output alphabets has been fully characterized. More recent results on AVCs can be found in \cite{chen2021strong,chen2022strong,yadav2022new,sarwate2010rateless,dey2024codes,pereg2018arbitrarily,pereg2019arbitrarily}. An interesting interpretation to differentiate the average and maximal error cases is the knowledge that the jammer has. A general assumption is that the jammer knows the communication standard because we are looking at public communication systems, and the communication standard is then generally known. However, the jammer can also know more, for example, the message, and then he has more options for choosing his jamming strategy. When the jammer does not know the transmitted message, the jamming strategy needs to be effective for most of the messages to corrupt the communication, which corresponds to the average error criterion. The diagram in Fig. \ref{fig:jammer diagram} illustrates the message transmission coding strategies depending on the jammer's knowledge. We consider both the average and maximal error criteria in our communication model, and then we see immediately that these different criteria must also be mapped exactly in the coding procedure. 
\begin{figure}
    \centering
    \includegraphics[scale=0.4]{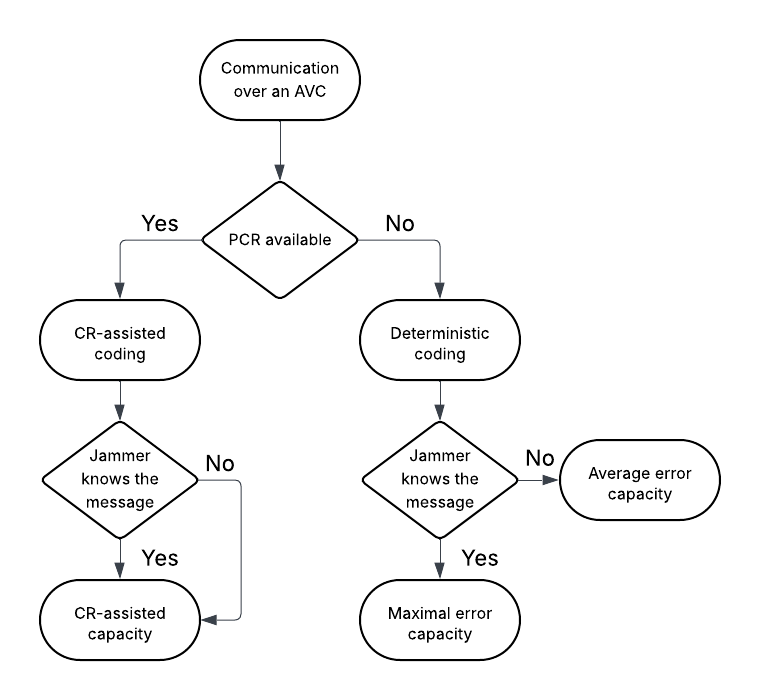}
    \caption{\footnotesize Message transmission coding strategies depending on the jammer's knowledge: The private correlated randomness (PCR) is available when the correlated randomness is only available at the encoder and decoder.}
    \label{fig:jammer diagram}
\end{figure}

This paper focuses on the robust joint message and state transmission problem. We consider a bi-static ISAC setting in which the state estimator coincides with the message receiver. Further, we assume the encoder has access to the channel state of interest. State transmission with different levels of encoder accessibility to the state has a wide range of applications in modern wireless communication systems. For example, in a watermarking system, the covering signal is the channel state that is noncausally available at the encoder, and the message is the watermark. The decoder is required to decode the watermark correctly and recover the covering signal with a given fidelity. In drone communication networks, drones are sent, as detectors, to detect the surrounding environment and communicate the results back to the base station. Based on the detection capability of the drone, it may detect the environment with some delay, resulting in a strictly causal observation of the channel state. The models studied in this paper capture practical and important applications of ISAC in 6G networks and other scenarios. Bounds of the discussed models without arbitrarily varying jamming attacks are provided in \cite{sutivong2005channel,kim2008state,choudhuri2013causal,ahmadipour2023integrated}. A weakness of drone communications is that they are always vulnerable due to the communication distance, limitations in transmission power, etc. Hence, it is of great importance to ensure the capability of the devices to communicate under unknown jamming. We consider channels with random states that need to be estimated, arbitrarily varying jamming under both average and maximal criteria, and different types of observation of the state sequences. As we discussed in the previous paragraph, the maximal error capacity of an AVC can be strictly smaller than the average error capacity, even when the channel is a point-to-point channel. We want to examine in this paper if the same property holds for the joint message and state transmission problem. Specifically, we consider the case where correlated randomness is not available. The main results of this paper are deterministic coding lower bounds of the capacity--distortion functions and some capacity-achieving special cases. It turns out that for the same level of distortion, the deterministic code capacity under the maximal error criterion can be strictly smaller than that under the average error criterion, which is in accordance with normal message communications over AVCs. For both strictly causal and noncausal cases, we provide tight bounds when the state sequences are required to be reconstructed losslessly. 

The rest of this paper is organized as follows. The definitions of the channel model and main results are provided in Section \ref{sec: models and results}. Special cases for binary-output channels are discussed in Section \ref{sec: binary case}. Section \ref{sec: conclusion} concludes this paper. All the proofs are provided in the appendices.

\section{Models and results}\label{sec: models and results}
Throughout this paper, random variable, sample value and its alphabet are denoted by capital, lowercase letter and calligraphic letter, respectively, e.g. $X$, $x$ and $\mathcal{X}$. Symbols $X^n$ and $x^n$ represent a random sequence and its sample value with length $n$. The distribution of a random variable $X$ is denoted by $Q_X$, and the joint distribution of a pair of random variables $(X,Y)$ is denoted by $Q_{XY}$. The expectation of a function of the random variable $X$ is written as $\mathbb{E}_X\left[ f(X) \right]$. The set of integers from $1$ to $N$ is denoted by $[1:N]$. 

For a given distribution $Q_X$, we define the $\delta-$typical set $\mathcal{T}^n_{Q_X,\delta}$ as
\begin{align}
    \mathcal{T}^n_{Q_X,\delta}=\{x^n\in\mathcal{X}^n: |N(x|x^n) - nQ_X(x)|\leq n\delta\}\;\;\forall x\in\mathcal{X},
\end{align}
where $N(x|x^n)$ is the number of occurrences of symbol $x$ in the sequence $x^n$. For simplicity, we use \emph{typical set} to denote $\delta-$typical set.  The definition of the joint typical set $Q_{XY}$ is defined similarly:
\begin{align}
    \mathcal{T}^n_{Q_{XY},\delta}=\{(x^n,y^n)\in\mathcal{X}^n\times\mathcal{Y}^n: |N(x,y|x^n,y^n) - nQ_{XY}(x,y)|\leq n\delta\}\;\;\forall (x,y)\in\mathcal{X}\times\mathcal{Y}.
\end{align}
The conditional typical set $\mathcal{T}^n_{Q_{XY},\delta}[x^n]$ given $x^n$ is
\begin{align}
    \mathcal{T}^n_{Q_{XY},\delta}[x^n] =\{y^n\in\mathcal{Y}^n: (x^n,y^n)\in \mathcal{T}^n_{Q_{XY},\delta}\}.
\end{align}
Properties of typical sets  can be found in \cite{el2011network}. Setting $\delta=0$ gives a special set of sequences
\begin{align}
    \mathcal{T}^n_{Q_X} = \{x^n\in\mathcal{X}^n: N(x|x^n) = nQ_X(x)\}\;\;\forall x\in\mathcal{X}.
\end{align}
We call $\mathcal{T}^n_{Q_X}$ a \emph{type set} and say the type of $x^n\in\mathcal{T}^n_{Q_X}$ is $Q_X$. The definitions of joint type sets and conditional type sets follow similarly. More properties of type sets can be found in \cite{csiszar2011information}.

To distinguish between the general distribution and type, for a given sequence $x^n$, we use $P_{x^n}$ to denote the \emph{type} of the sequence such that $P_{x^n}(x)$ is the relative frequency of the symbol $x\in\mathcal{X}$ in sequence $x^n.$ On some occasions, we also define a dummy random variable $X$ according to $x^n$ such that the distribution $P_X$ of $X$ is the type $P_{x^n}$. In this case, the following expressions are equivalent:
\begin{align}
    H(P_{x^n})=H(x^n)=H(X),
\end{align}
where $X$ is the dummy random variable defined by $x^n$.
Similarly, for a pair of sequences $(x^n,y^n)$ with joint type $P_{x^n,y^n}$, the following expressions are equivalent:
\begin{align}
    I(P_{x^n};P_{y^n|x^n})=I(X;Y)=I(x^n;y^n),
\end{align}
where $(X,Y)$ are dummy random variables defined by $(x^n,y^n)$. Generally speaking, we use $Q_X\in\mathcal{P}(\mathcal{X})$ to denote a general distribution of $X$ and use $P_X$ to denote the distribution defined by the type of a given sequence $x^n$. We use $\overset{\cdot}{=}$ to denote the equality in the exponential scale, i.e. $a\overset{\cdot}{=}b$ implies $\lim_{n\to\infty}\frac{1}{n}\log \frac{a}{b}\to 0$.

Throughout this paper, we consider random variables with finite alphabet sizes.
\subsection{Channel Model}
\begin{figure}[h]
    \centering
    \includegraphics[scale=0.6]{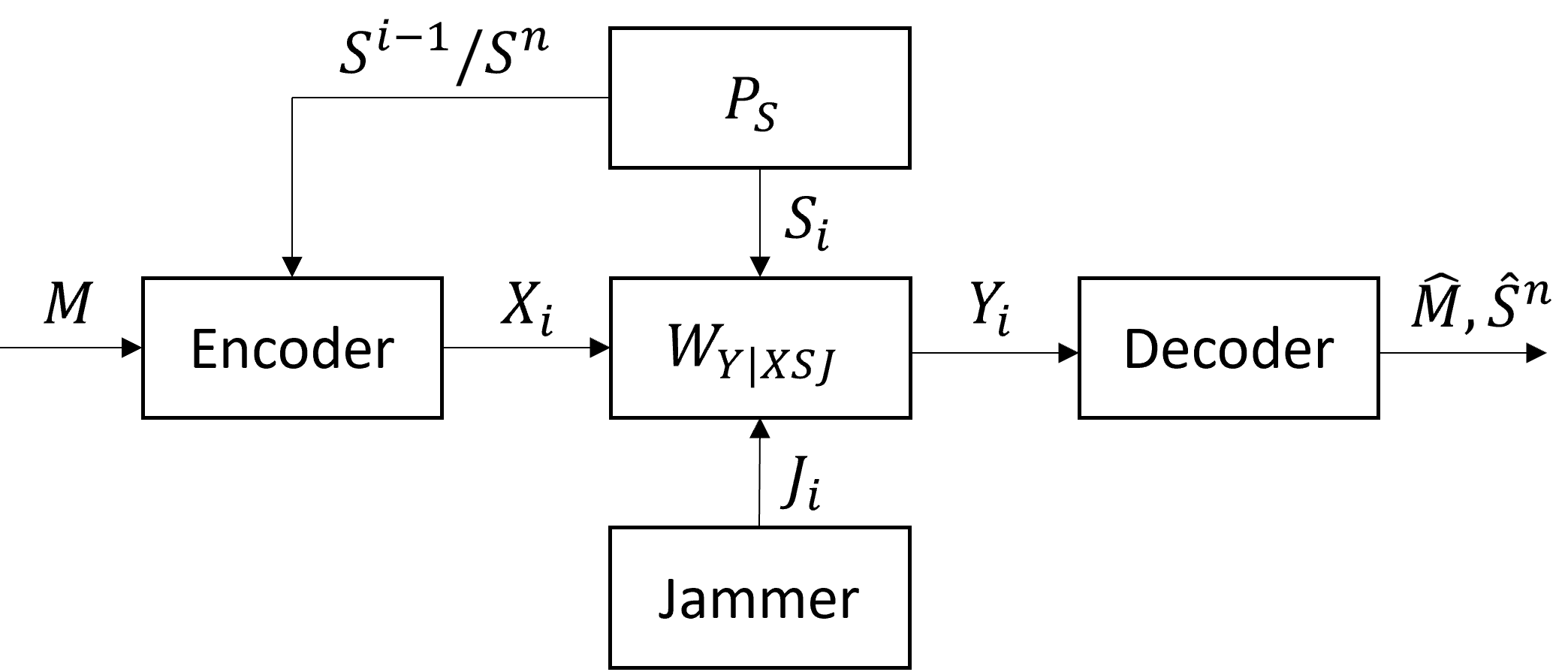}
    \caption{{\footnotesize Joint message and state communication under arbitrarily varying jamming}}
    \label{fig: model}
\end{figure}
As shown in Fig. \ref{fig: model}, the model considered in this paper is a discrete memoryless channel with random state $S$ and a jammer. The random state $S$ follows a fixed distribution $Q_S$, while the jammer produces the jamming sequences $J^n$ with arbitrary distributions $Q_{J^n}$. The components of $J^n$ may not be independent of each other, and the distribution $Q_{J^n}$ is not necessarily a product law.
The encoder tries to send the message $M$ to the receiver and has access to the random state sequence $S^n$ in strictly causal or noncausal manners. The decoder aims to decode the message $M$ and produce a sequence $\hat{S}^n$, which is an estimation of the state sequence $S^n$ based on its observation. The communication quality is measured by the decoding error, and the estimation quality is measured by a given distortion function $d:\mathcal{S}\times\hat{\mathcal{S}}\to [0,+\infty)$ such that $\frac{1}{n}d(S^n,\hat{S}^n)=\frac{1}{n}\sum_{i=1}^n d(S_i,\hat{S}_i).$

Both average and maximal error criteria are considered in this paper and the corresponding definitions will be given in the rest of this section. We assume that the jammer always knows the communication protocol. In practice, this is always the case because we are looking at public communication systems where the communication protocol is standardized. What is important for us is that the jammer does not know the random state of the channel, which is why we can also average over the state when using the maximum error criterion. 

\begin{definition}\label{def: strictly causal code def}
    A code $(n,R)$ used for joint message and state communication with strictly causal observation consists of
    \begin{itemize}
        \item A message set $\mathcal{M}=[1:2^{nR}]$;
        \item A sequence of encoders $f_i:\mathcal{M}\times \mathcal{S}^{i-1} \to \mathcal{X}_i,i=1,...,n$;
        \item A message decoder $g: \mathcal{Y}^n\to \mathcal{M}$;
        \item A state estimator $h:\mathcal{Y}^n\to\hat{\mathcal{S}}^n.$
    \end{itemize}
    The strictly-causal case can be considered as  a special case of coding with delayed state observation with delay $d=1$. For more general cases, the sequence of encoders can be defined by $f_i:\mathcal{M} \to \mathcal{X}_i$ for $i\leq d$ and $f_i:\mathcal{M}\times \mathcal{S}^{i-d} \to \mathcal{X}_i$ for $d < i \leq n.$
\end{definition}

\begin{definition}
    A code $(n,R)$ used for joint message and state communication with noncausal observation consists of
    \begin{itemize}
        \item A message set $\mathcal{M}=[1:2^{nR}]$;
        \item An encoder $f:\mathcal{M}\times \mathcal{S}^n \to \mathcal{X}^n$;
        \item A message decoder $g: \mathcal{Y}^n\to \mathcal{M}$;
        \item A state estimator $h:\mathcal{Y}^n\to\hat{\mathcal{S}}^n.$
    \end{itemize}
\end{definition}

This paper considers the joint message and state communication under both average and maximal error criteria. It is well known that the capacity of arbitrarily varying channels can be different under different error criteria. In the following, we give definitions of the achievable rate--distortion pair and introduce the difference between these two cases. 

\subsection{Average Error Criterion}

\begin{definition}\label{def: achievability of average error}
    A rate--distortion pair $(R,D)$ is achievable with strictly causal state information under the average error criterion if for any $\epsilon>0$, there exists a sufficiently large number $N$ such that for any $n\geq N$ there exists an $(n,R)$ strictly causal code such that $\forall j^n\in\mathcal{J}^n,$
    \begin{align}
        &e_{a,s}(j^n) = \frac{1}{|\mathcal{M}|}\sum_{m=1}^{|\mathcal{M}|}e_s(m,j^n)\leq \epsilon, \\
        &\frac{1}{n}\mathbb{E}[d(S^n,\hat{S}^n)]\leq D.
    \end{align}
    where 
     \begin{align*}
         e_s(m,j^n) = \sum_{y^n\notin g^{-1}(m)}\sum_{s^n}Q^n_S(s^n)\prod_{i=1}^n W(y_i|f_i(m,s^{i-1}),s_i,j_i).
     \end{align*}
    The capacity--distortion function $C_{a,s}(D)$ is the supremum of the rate $R$ such that $(R,D)$ is achievable under the average error criterion. 
    
    Similarly, a rate--distortion pair $(R,D)$ is achievable with noncausal state information under the average error criterion if for any $\epsilon>0$, there exists a sufficiently large number $N$ such that for any $n\geq N$ there exists an $(n,R)$ noncausal code such that $\forall j^n\in\mathcal{J}^n,$
    \begin{align}
        &e_{a,n}(j^n) = \frac{1}{|\mathcal{M}|}\sum_{m=1}^{|\mathcal{M}|}e_n(m,j^n)\leq \epsilon, \\
        &\frac{1}{n}\mathbb{E}[d(S^n,\hat{S}^n)]\leq D.
    \end{align}
    where 
     \begin{align*}
         e_n(m,j^n) = \sum_{y^n\notin g^{-1}(m)}\sum_{s^n}Q^n_S(s^n)\prod_{i=1}^n W(y_i|x_i,s_i,j_i),
     \end{align*}
     where $x^n = f(m,s^n)$.
    The capacity--distortion function $C_{a,n}(D)$ is the supremum of the rate $R$ such that $(R,D)$ is achievable under the average error criterion. 
\end{definition}

The sufficient and necessary condition for an AVC to have a positive average error capacity is the nonsymmetrizability\cite{csiszar1988capacity}\cite{ericson1985exponential}. The following definitions are the symmetrizable conditions of arbitrarily varying multiple access channels (AVMACs) and their extensions.
\begin{definition}[Symmetrizability]\label{def: symmetrizability}
    A discrete memoryless channel with random states and arbitrarily varying jamming is said to be
    \begin{itemize}
        \item \textbf{symmetrizable}-$\mathcal{X}\times \mathcal{S}$ if there exists some distributions $T_{J|XS}$ such that
            \begin{align*}
            &\sum_{j}W_{Y|XSJ}(y|x,s,j)T(j|x',s') \\
            &\quad = \sum_{j}W_{Y|XSJ}(y|x',s',j)T(j|x,s) \;\;\;\text{$\forall y,x,s,x',s'.$}
        \end{align*}
        \item \textbf{symmetrizable}-$\mathcal{X}$ if there exists some distributions $T_{J|X}$ such that
            \begin{align*}
            &\sum_{j}W_{Y|XSJ}(y|x,s,j)T(j|x') \\
            &\quad= \sum_{j}W_{Y|XSJ}(y|x',s,j)T(j|x) \;\;\;\text{$\forall y,x,s,x'.$}
        \end{align*}
        \item \textbf{symmetrizable}-$\mathcal{S}$ if there exists some distributions $T_{J|S}$ such that
            \begin{align*}
            &\sum_{j}W_{Y|XSJ}(y|x,s,j)T(j|s') \\
            &\quad = \sum_{j}W_{Y|XSJ}(y|x,s',j)T(j|s) \;\;\;\text{$\forall y,x,s,s'.$}
        \end{align*}
        \item \textbf{symmetrizable}-$\mathcal{X}|\mathcal{S}$ if there exists some distributions $T_{J|XS}$ such that
            \begin{align*}
            &\sum_{j}W_{Y|XSJ}(y|x,s,j)T(j|x',s) \\
            &\quad = \sum_{j}W_{Y|XSJ}(y|x',s,j)T(j|x,s) \;\;\;\text{$\forall y,x,x',s.$}
        \end{align*}
        \item \textbf{symmetrizable}-$\mathcal{S}|\mathcal{X}$ if there exists some distributions $T_{J|XS}$ such that
            \begin{align*}
            &\sum_{j}W_{Y|XSJ}(y|x,s,j)T(j|x,s') \\
            &\quad = \sum_{j}W_{Y|XSJ}(y|x,s',j)T(j|x,s) \;\;\;\text{$\forall y,x,s,s'.$}
        \end{align*}
    \end{itemize}
\end{definition}
For the point-to-point case, the definition of the symmetrizable-$\mathcal{X}$ AVC follows from removing $S$ from the above symmetrizable-$\mathcal{X}$ channels \cite{csiszar1988capacity}. If we consider $X$ and $S$ as two inputs of the channel $W_{Y|XSJ}$ and $J$ as the jamming symbol,  the first three definitions are exactly the symmetrizability of an arbitrarily varying multiple access channel\cite{gubner1990deterministic}. We will see later that the symmetrizability of AVMAC will be useful for lossless state reconstruction.
\begin{lemma}\label{lem: nonsymmetrizable u|x inequality}
    If the channel $Q_{Y|XSJ}$ is nonsymmetrizable-$\mathcal{S}|\mathcal{X}$, there exists $\zeta>0$ such that for each pair of channels $T_1:\mathcal{X}\times\mathcal{S}\to\mathcal{J}$ and $T_2:\mathcal{X}\times\mathcal{S}\to\mathcal{J}$,
    \begin{align*}
        &\max_{x,s,s',y}|\sum_{j}T_1(j|x,s)Q_{Y|XSJ}(y|x,s',j)\\
        &\quad\quad\quad\quad -\sum_{j}T_2(j|x,s')Q_{Y|XSJ}(y|x,s,j)|\geq \zeta,
    \end{align*}
\end{lemma}
The proof is given in Appendix \ref{sec: proof of lem: nonsymmetrizable u|x inequality}. In the following, we provide an example that satisfies all symmetrizable conditions defined in Definition \ref{def: symmetrizability}.
\begin{example}
    Consider a channel 
    \begin{align}
        Y = X \oplus S \oplus J,
    \end{align}
    where $\mathcal{Y}=\mathcal{X}=\mathcal{S}=\mathcal{J}=\{0,1\}$, $\oplus$ is the modulo 2 addition. This channel satisfies all symmetrizable conditions in Definition \ref{def: symmetrizability}. We first verify that it is symmetrizable-$\mathcal{X}\times\mathcal{S}$. Fix $(x,s)=(0,0)$ and $(x',s')=(0,1)$. By considering the cases that $y=0$ and $y=1$, the channel $T_{J|XS}$ that symmetrizes $W_{Y|XSJ}$ satisfies
    \begin{align}
        &T_{J|XS}(0|0,0) = T_{J|XS}(1|0,1),\\
        &T_{J|XS}(1|0,0) = T_{J|XS}(0|0,1).
    \end{align}
    Choosing $(x',s')=(1,0)$, it satisfies
    \begin{align}
        &T_{J|XS}(0|0,0) = T_{J|XS}(1|1,0),\\
        &T_{J|XS}(1|0,0) = T_{J|XS}(0|1,0).
    \end{align}
    Choosing $(x',s')=(1,1)$, it satisfies
    \begin{align}
        &T_{J|XS}(0|0,0) = T_{J|XS}(0|1,1),\\
        &T_{J|XS}(1|0,0) = T_{J|XS}(1|1,1).
    \end{align}
    Hence, the desired channel
    \begin{align}
        T_{J|XS} = \begin{bmatrix}
            a & 1-a \\
            1-a & a \\
            1-a & a\\
            a & 1-a
        \end{bmatrix}
    \end{align}
    for $a\in [0,1]$. Similarly, we can show that the channel is symmetrizable-$\mathcal{X}$ (symmetrizable-$\mathcal{S}$) by setting $T_{J|X} (T_{J|S})$ to be a binary symmetric channel and the channel is symmetrizable-$\mathcal{X}|\mathcal{S}$ by setting 
    \begin{align}
        T_{J|XS}=\begin{bmatrix}
            a & 1-a \\
            b & 1-b \\
            1-a & a\\
            1-b & b
        \end{bmatrix}
    \end{align}
    for $a\in[0,1],b\in[0,1]$. The symmetrizable-$\mathcal{S}|\mathcal{X}$ case follows similarly.
\end{example}
In Example \ref{example: nonsymmetrizable U|X}, we further provide some examples such that the channel is nonsymmetrizable.
 In \cite{csiszar1988capacity}, the capacity of nonsymmetrizable AVCs under the average error criterion is provided.
 \begin{theorem}{\cite[Theorem 1]{csiszar1988capacity}}
     For the average error criterion, the deterministic coding capacity $C_a>0$ if and only if the AVC is nonsymmetrizable. If $C_a>0$, then
     \begin{align}
         C_a = \max_{Q_X}\min_{Q_J} I(X;Y),
     \end{align}
     with the joint distribution $Q_XQ_JW_{Y|XJ}.$
 \end{theorem}

\subsection{Maximal Error Criterion}
\begin{definition}
     A rate--distortion pair $(R,D)$ is achievable with strictly causal state information under the maximal error criterion if for any $\epsilon>0$, there exists a sufficiently large number $N$ such that for any $n\geq N$ there exists an $(n,R)$ strictly causal code such that $\forall j^n\in\mathcal{J}^n,$
    \begin{align}
        &e_{m,s}(j^n) = \max_{m\in\mathcal{M}}e_s(m,j^n)\leq \epsilon, \\
        &\frac{1}{n}\mathbb{E}[d(S^n,\hat{S}^n)]\leq D.
    \end{align}
    where 
     \begin{align*}
         e_s(m,j^n) = \sum_{y^n\notin g^{-1}(m)}\sum_{s^n}Q^n_S(s^n)\prod_{i=1}^n W(y_i|f_i(m,s^{i-1}),s_i,j_i).
     \end{align*}
    The capacity--distortion function $C_{m,s}(D)$ is the supremum of the rate $R$ such that $(R,D)$ is achievable under the maximal error criterion. 
    
    Similarly, a rate--distortion pair $(R,D)$ is achievable with noncausal state information under the maximal error criterion if for any $\epsilon>0$, there exists a sufficiently large number $N$ such that for any $n\geq N$ there exists an $(n,R)$ noncausal code such that $\forall j^n\in\mathcal{J}^n,$
    \begin{align}
        &e_{m,n}(j^n) = \max_{m\in\mathcal{M}}e_n(m,j^n)\leq \epsilon, \\
        &\frac{1}{n}\mathbb{E}[d(S^n,\hat{S}^n)]\leq D.
    \end{align}
    where 
     \begin{align*}
         e_n(m,j^n) = \sum_{y^n\notin g^{-1}(m)}\sum_{s^n}Q^n_S(s^n)\prod_{i=1}^n W(y_i|x_i,s_i,j_i),
     \end{align*}
     where $x^n = f(m,s^n)$.
    The capacity--distortion function $C_{m,n}(D)$ is the supremum of the rate $R$ such that $(R,D)$ is achievable under the maximal error criterion. 
\end{definition}
For the maximal error criterion, capacity results are given in \cite{ahlswede1980method,csiszar1981capacity}. \cite{ahlswede1980method} introduces a graph with the vertex set $\mathcal{X}$ for a point-to-point AVC $\mathcal{W}$, denoted by $G(\mathcal{W}).$ In the graph, $x_1$ and $x_2$ are connected by an edge, denoted by $x_1 \overset{\mathcal{W}}{\sim}x_2$ iff there exist distributions $Q_1$ and $Q_2$ on $\mathcal{J}$ such that
\begin{align}
    \label{def: maximal error symmetrizability}\sum_{j\in\mathcal{J}}W(y|x_1,j)Q_1(j) = \sum_{j\in\mathcal{J}}W(y|x_2,j)Q_2(j) \;\;\text{for all $y\in\mathcal{Y}$.}
\end{align}
The capacity of arbitrarily varying channels under the maximal error criterion is positive if and only if the graph is not a complete graph. 

It should be noted that the zero-capacity condition for AVCs under the maximal error criterion is weaker than that under the average error criterion. For a  symmetrizable-$\mathcal{X}$ AVC $\mathcal{W}$, its graph $\mathcal{G}(\mathcal{W})$ is always complete by setting $Q_1(j)=T_{J|X}(j|x_2),Q_2(j)=T_{J|X}(j|x_1)$. However, the other direction does not always hold since the distributions $Q_1$ and $Q_2$ can be different for the maximal error criterion. Therefore, the capacity of an AVC under the maximal error criterion is always upper bounded by that under the average error criterion.

Furthermore, a pair of input symbols $(x_1,x_2)$ is a pair of isolated vertices in the graph if
\begin{align}
    \label{def: maximal error nonsymmetrizability}\min_{Q_1,Q_2\in\mathcal{P}(\mathcal{J})}\sum_{y}\Big|\sum_{j\in\mathcal{J}}W(y|x_1,j)Q_1(j) -\sum_{j\in\mathcal{J}}W(y|x_2,j)Q_2(j)\Big| >0.
\end{align}
The following example provides an AVC whose graph is not complete.
\begin{example}
    Consider the following AVC with $\mathcal{Y}=\mathcal{X}=\{0,1,2\}$ and $\mathcal{J}=\{0,1\}$:
    \begin{align}
        W_{Y|XJ=0} = 
        \begin{bmatrix}
        0 & 1 & 0 \\
        1 & 0 & 0 \\
        0 & 0 & 1
    \end{bmatrix},\;\;\;
    W_{Y|XJ=1} = 
    \begin{bmatrix}
        0 & 0 & 1 \\
        0 & 1 & 0 \\
        \frac{1}{2} & 0 & \frac{1}{2}
    \end{bmatrix}.
    \end{align}
    It follows that $(x,x')=(0,1)$ and $(x,x')=(0,2)$ are connected but $(x,x')=(1,2)$ are isolated. To see this, by choosing $y=2,$ if $(x,x')=(1,2)$ are connected, there exist distributions $Q_{1}$ and $Q_2$ on $\mathcal{J}$ such that
    \begin{align}
        0\cdot Q_1(0) + 0 \cdot Q_1(1) = 1 \cdot Q_2(0) + \frac{1}{2}Q_2(1).
    \end{align}
    Obviously, there does not exist such a distribution $Q_2$ and hence $x=1$ and $x'=2$ are isolated.
\end{example}
Define $C(Q_X) := \min_{Q_{J|X}} I(X;Y)$, where $Q_X$ is the distribution of $X$. It is proved in \cite{ahlswede1980method} that when \eqref{def: maximal error nonsymmetrizability} holds for all $x,x'\in\mathcal{X}$, the maximal error capacity is
\begin{align}
    C_m = \max_{Q_X}C(Q_X) = \max_{Q_X} \min_{Q_{J|X}}I(X;Y).
\end{align}
Such a condition is relaxed in \cite{csiszar1981capacity}, and a rate $R$ given input distribution $Q_X$ is achievable if 
\begin{align}
    R \leq \min(C(Q_X),D(Q_X)),
\end{align}
where
\begin{align}
    D(Q_X):=\min_{\substack{Q_{X\widetilde{X}}:\\Q_X=Q_{\widetilde{X}},\\Pr\{X \overset{\mathcal{W}}{\sim} \widetilde{X}\}=1}} I(X;\widetilde{X}).
\end{align}

For a channel $\mathcal{W}$ with random states and arbitrarily varying jamming, we define its graph with respect to $\mathcal{X}$ by $\mathcal{G}_{\mathcal{X}}(\mathcal{W})$ with the following definitions.
\begin{definition}
    A pair of input symbols $(x_1,x_2)$ are isolated in the graph if
    \begin{align}
    \label{def: isolated x1 x2}\min_{Q_{1},Q_{2}\in\mathcal{P}(\mathcal{J})}\sum_{y}\Big|\sum_{j\in\mathcal{J}}Q(y|x_1,j)Q_{1}(j) -\sum_{j\in\mathcal{J}}Q(y|x_2,j)Q_{2}(j)\Big| >0,
\end{align}
where $Q_{Y|XJ}=\sum_s W_{Y|XSJ}Q_S$.
It is a complete-$\mathcal{X}$ graph if there does not exist a pair $(x_1,x_2)$ that is isolated.
\end{definition}

Reliable communication is not possible over an AVC $\mathcal{W}$ when $\mathcal{G}(\mathcal{W})$ is a complete graph. The following lemma, although not used in this paper, was first proved in \cite{kiefer1962channels} as the necessary condition for the maximal error of an AVC to be positive. Here, we prove it using the terminology defined in this paper. Suppose there exists a codebook $\{x^n(m):m\in\mathcal{M}\}$ and its index set $\mathcal{M}$, an encoding function $f$ that defines the rule of how $x^n(m)$ is selected, and a decoding function $g: \mathcal{Y}^n\to\mathcal{M}$. Define 
    \begin{align}
        e(m,j^n) = \sum_{y^n\notin g^{-1}(m)}W^n(y^n|x^n(m),j^n).
    \end{align}
    The following lemma demonstrates that when the graph of the channel is a complete graph, the decoder cannot distinguish between any two codewords in the codebook, regardless of the codebook's rate.
\begin{lemma}\label{lem: necessity of complete graph}
    If the graph of an AVC is complete-$\mathcal{X}$, there does not exist an encoding/decoding strategy such that
    \begin{align}
        \max_{j^n}e(m,j^n) \leq \frac{1}{2}.
    \end{align}
    for any $m\in\mathcal{M}$ when $|\mathcal{M}|\geq 2$.
\end{lemma}
The proof is given in Appendix \ref{sec: proof of lem: necessity of complete graph}.

\subsection{Main Results: Average Error Criterion}
Fix a joint distribution $Q_{XSUY|J}=Q_XQ_SQ_{U|XS}W_{Y|XSJ}$ and let
\begin{align}
    &Q_{Y|XJ}=\sum_s W_{Y|XSJ}Q_{S},\\
    &Q_{Y|UXJ}=\sum_{s}W_{Y|XSJ}Q_{S|XU},
\end{align}
where $Q_{S|UX}=Q_{X}Q_SQ_{U|XS}/(\sum_s Q_{X}Q_SQ_{U|XS})$.
Define an indicator function
\begin{align}
    \mathbb{I}_{\mathcal{U}|\mathcal{X}}(Q_{Y|UXJ}) = \left\{
        \begin{aligned}
            &1\;\;\text{if $Q_{Y|UXJ}$ is nonsymmetrizable-$\mathcal{U}|\mathcal{X}$},\\
            &0\;\;\text{otherwise.}
        \end{aligned}
    \right.
\end{align}
\begin{theorem}\label{the: achievability of strictly causal}
   For the deterministic coding under average error criterion and strictly causal observation, the capacity--distortion function of the joint message and state transmission over arbitrarily varying channels with random states $W_{Y|XSJ}$ such that $Q_{Y|XJ}$ is nonsymmetrizable-$\mathcal{X}$ is lower bounded by
    \begin{align*}
        C_{a,s}(D) \geq \max \left(\min_{Q_J}I_{Q_J}(X;Y) + \min_{Q_J} I_{Q_J}(U;Y|X) - I(U;S|X)\right),
    \end{align*}
    where the maximization is taken over all distributions $Q_X,Q_{U|XS}$ and deterministic functions $h:\mathcal{X}\times\mathcal{U}\times\mathcal{Y}\to\hat{\mathcal{S}}$ such that $\mathbb{I}_{\mathcal{U}|\mathcal{X}}(Q_{Y|UXJ})=1$ and $\max_{Q_J}\mathbb{E}[d(S,h(X,U,Y))]\leq D$.
\end{theorem}
\begin{remark}
    Here we consider some degenerate cases:
    \begin{itemize}
        \item When $\mathcal{S}=1$, by setting $U=\emptyset$, the result reduces to the capacity of a point-to-point AVC.
        \item When $\mathcal{J}=1$, the model becomes strictly causal state communication \cite{choudhuri2013causal} and our result reduces to the optimal capacity--distortion function \cite[Theorem 2]{choudhuri2013causal}. Furthermore, we can also write the inverse of the capacity--distortion function, namely the distortion--capacity function, and recover the pure state reconstruction bound when $R=0$ \cite[Theorem 1]{choudhuri2013causal}.  
    \end{itemize}
\end{remark}
The proof is provided in Appendix \ref{sec: proof of achievability of strictly causal}. The first part $\min_{Q_J}I_{Q_J}(X;Y)$ is the total rate that one can reliably transmit through this channel with a given input distribution under any possible jamming sequence, which is the direct result of \cite{csiszar1988capacity} based on the method of type. The remaining part is the rate of the description of the channel states and is encoded in a Wyner-Ziv coding manner with $(S^n,Y^n)$ being considered as a pair of correlated sources. We use a block Markov coding scheme, and in each block the transmitted codeword is determined by two parts: the message to be sent and the Wyner-Ziv bin index of the description of the channel state sequence from the previous block. The message and the description are decoded based on two different channel output observations, which results in two minimizations over the jamming distributions in the achievability result.

The main difficulty in building the coding scheme is the fact that the description $U^n$ relates to the channel output only through its bin index. However, due to the law of large numbers, the connection of $U^n$ and $Y^n$ can still be simulated by a new AVC $Q_{Y|UXJ}$, and the bin index of $U^n$, after taking the average over all $S^n$, is almost uniformly distributed. Hence, it is sufficient to construct another code for the AVC $Q_{Y|UXJ}$ under the average error criterion, and this completes the proof.

It should be noted that the capacity is positive only when the average channel $Q_{Y|XJ}$ is nonsymmetrizable-$\mathcal{X}$. It turns out that the symmetrizability of $S$ does not affect the decoding of $X$ since the error probability takes the average over all $S^n$ in both average and maximal error cases. Then, it follows from the results of \cite{csiszar1988capacity}\cite{ericson1985exponential} that the capacity of an AVC $Q_{Y|XJ}$ is positive if and only if the channel is nonsymmetrizable-$\mathcal{X}$. When the conditions in the theorem do not hold, the codeword $X^n$ cannot be transmitted reliably, and hence, no message and state description can be conveyed to the receiver. However, the estimator can still estimate the state $S^n$ under some level $D$ since the channel output $Y^n$ is correlated to the state sequence, and the estimation function is solely based on $Y^n$.

On the other hand, the state communication restricts the choice of input distributions such that the resulting channel $Q_{Y|UXJ}$ is nonsymmetrizable-$\mathcal{U}|\mathcal{X}$. It is a constraint more strict than nonsymmetrizable-$\mathcal{U}$ since a nonsymmetrizable-$\mathcal{U}|\mathcal{X}$ is always nonsymmetrizable-$\mathcal{U}$ but the converse does not necessarily hold. It is also stronger than nonsymmetrizable-$\mathcal{X}\times\mathcal{U}$ since a symmetrizable-$\mathcal{X}\times\mathcal{U}$ channel is always symmetrizable-$\mathcal{U}|\mathcal{X}$ by setting $x'=x$ in the definition. When the channel is symmetrizable-$\mathcal{S}|\mathcal{X}$, the resulting channel $Q_{Y|XUJ}$ may not be symmetrizable-$\mathcal{U}|\mathcal{X}$ by carefully choosing the distribution $Q_{U|XS}$. This can be shown by the following example.
\begin{example}\label{example: nonsymmetrizable U|X}
    Suppose the channel satisfies
    \begin{align}
        Y= X + S + J
    \end{align}
    with alphabets $\mathcal{Y}=\{0,1,2,3,4\}, \mathcal{X}=\{0,1,2\},\mathcal{S}=\mathcal{J}=\{0,1\}.$ It is easy to verify that the channel is nonsymmetrizable-$\mathcal{X}\times\mathcal{S}$ and nonsymmetrizable-$\mathcal{X}$. To see this, by Definition \ref{def: symmetrizability} and setting $y=0$, we have
    \begin{align}
        P_{J|XS}(0|x',s') = 0 \;\;\text{$\forall (x',s')\neq (0,0)$}.
    \end{align}
    By setting $y=4$ we then have
    \begin{align}
        P_{J|XS}(1|x',s') = 0 \;\;\text{$\forall (x',s')\neq (2,1)$}.
    \end{align}
    It follows that $P_{J|XS}(0|x,s)=P_{J|XS}(1|x,s)=0$ for all $(x,s)\neq (0,0)$ and $(x,s)\neq (2,1)$. Hence, $P_{J|XS}$ is not a distribution and the channel is nonsymmetrizable-$\mathcal{X}\times\mathcal{S}$. By similar arguments, it follows that the channel is nonsymmetrizable-$\mathcal{X}$ but symmetrizable-$\mathcal{S}$ and symmetrizable-$\mathcal{S}|\mathcal{X}$. Now we choose the distribution $P_{U|XS}$ to be $U=X+S$. Then the alphabet $\mathcal{U}=\{0,1,2,3\}$.

    Again we set $y=0$, and this implies $x=s=j=0$ and hence $u=0$. By Definition \ref{def: symmetrizability}, if the channel $Q_{Y|XUJ}$ is symmetrizable-$\mathcal{U}|\mathcal{X}$,
    \begin{align}
        P_{J|U'X}(J=0|u',X=0) = 0 \;\;\text{$\forall u'\neq 0$}.
    \end{align}
    Setting $y=2,x=0$ implies $s=j=1$ and hence $u=1$, which gives
    \begin{align}
        P_{J|U'X}(J=1|u',X=0) = 0 \;\;\text{$\forall u'\neq 1$}.
    \end{align}
    Combining the above two equations together gives $P_{J|UX}(J=0|u,X=0)=P_{J|UX}(J=1|u,X=0)=0$ for $u=2,3$. Therefore, there does not exist a distribution $P_{J|UX}$ such that the channel is symmetrizable-$\mathcal{U}|\mathcal{X}$.
\end{example}

Considering the strictly causal case as coding with delayed state observation $d=1$ as defined in Definition \ref{def: strictly causal code def}, from the achievability proof, the proposed coding scheme in fact works for any $1\leq d \leq n$.
For the lossless state communication case, we have the following result.
\begin{corollary}\label{coro: strictly causal lossless average error}
    For the channel considered in Theorem \ref{the: achievability of strictly causal}, if it is also nonsymmetrizable-$\mathcal{S}|\mathcal{X}$, the deterministic coding capacity of the joint message and lossless state communication over this channel under the average error criterion is
    \begin{align}
        \label{eq: lossless capacity}C_{a,s} \geq \max_{Q_X} \min_{Q_J,Q_J'} [I_{Q_J}(X;Y) - H_{Q_J'}(S|X,Y)]^+.
    \end{align}
    where $[a]^+ := \max\{a,0\}.$ When the state observation has a delay $d=n$, the bound is tight.
\end{corollary}
The proof is given in Appendix \ref{sec: proof of coro: strictly causal lossless average error}.
\begin{remark}
    It might be surprising that the symmetrizability of $S$ does not matter as long as one can find an encoding strategy such that $Q_{Y|UXJ}$ is not symmetrizable-$\mathcal{U}|\mathcal{X}$. However, this is not the case for the lossless reconstruction.
\end{remark}

\begin{remark}
    The lower bound expression in \eqref{eq: lossless capacity} is the difference between two terms, in which the first one captures the influence of the jammer on the message transmission task, and the entropy term captures the cost of lossless state reconstruction under the jamming. To make sure the coding scheme fulfills the average decoding error and lossless reconstruction constraints independent of the jamming strategy, we have to consider the worst cases for both tasks independently, which means the minimal communication rate and the maximal lossless reconstruction cost. 
\end{remark}

\begin{corollary}\label{coro: noncausal average error lower bound}
    The capacity--distortion function for the noncausal case using deterministic coding under the average error criterion is lower bounded by
    \begin{align}
        C_{a,n}(D) \geq \max\{\min_{Q_J} I(U;Y) - I(U;S)\},
    \end{align}
    where the maximum is taken over all distributions $Q_{U|S},Q_{X|US}$ and deterministic functions $h: \mathcal{U}\times\mathcal{Y}\to\hat{\mathcal{S}}$ such that $Q_{Y|UJ}$ is nonsymmetrizable-$\mathcal{U}$ and $\max_{Q_J}\mathbb{E}[d(S,h(U,$ $Y))]\leq D$.
\end{corollary}
Taking average over all $s^n$, the probability of each $U^n$ in the codebook being selected is asymptotically equal to the uniform distribution as shown by \eqref{ine: probability of LK}. Hence, coding schemes for normal AVCs under the average error criterion are sufficient. The encoder generates a subcodebook for each message by random selection from a fixed type set $T^n_{Q_U}$, where $Q_U=\sum_s Q_{U|S}Q_S$. The covering error probability can be bounded by the Chernoff bound and is double exponentially small with respect to $n$, and the remaining proof follows the analysis in \cite{csiszar1988capacity} for the channel $Q_{Y|UJ}.$

For the pure lossless state communication, we have the following tight bound.
\begin{corollary}\label{coro: noncausal pure lossless state com}
    For the pure lossless state communication over a channel that is  nonsymmetrizable-$\mathcal{X}\times\mathcal{S}$, the sequence $S^n$ can be constructed losslessly using deterministic coding if and only if 
    \begin{align}
        \label{ine: condition for noncausal pure lossless state communication}H(S) \leq \min_{Q_J}\max_{Q_{X|S}} I(X,S;Y).
    \end{align}
\end{corollary}
\begin{remark}
    When $\mathcal{J}=1$, the problem reduces to the pure state reconstruction on a channel with random states. Corollary \ref{coro: noncausal pure lossless state com} reduces to \cite[Corollary 1]{choudhuri2012non}.  
\end{remark}
The proof is given in Appendix \ref{sec: proof of coro: noncausal pure lossless state com}.
For the pure lossless state communication case, by the law of large numbers, it is sufficient to consider all the typical sequences $S^n\in\mathcal{T}^n_{Q_S,\delta}$. The left-hand side is the rate we need to describe all the typical sequences, and the state communication will succeed if the channel can reliably communicate so much information. At first glance, the result follows by setting $U=(X,S)$ in Corollary \ref{coro: noncausal average error lower bound} with $0$ message rate, which means the encoder generates $X^n$ for each typical $S^n$ and input the pair $(X^n,S^n)$ as the codeword to the channel. Then, the right-hand side of \eqref{ine: condition for noncausal pure lossless state communication} is the channel capacity, and the nonsymmetrizable-$\mathcal{U}$ in Corollary \ref{coro: noncausal average error lower bound} becomes nonsymmetrizable-$\mathcal{X}\times\mathcal{S}$. However, the proof of achievability is not straightforward. This is because in the construction of the coding scheme, all the codewords $(X^n,S^n)$ should have the same type, while the sequence $S^n$ is generated according to $Q_S^n.$ To solve this problem, we construct a code for each type of $S^n$, and communicate the type to the receiver using a prefixed code. Due to the asymptotic equipartition property of the typical sequence, constructing average error codes is sufficient.

Given \eqref{ine: condition for noncausal pure lossless state communication}, the nonsymmetrizable-$\mathcal{X}\times\mathcal{S}$ condition is sufficient and necessary for the lossless communication of the state sequence. When the condition does not hold, $\eqref{ine: condition for noncausal pure lossless state communication}$ can still apply but the lossless reconstruction of the state at the receiver is not possible.

\subsection{Main Results: Maximal Error Criterion}

For distributions $Q_{SU}Q_{X|US}W_{Y|XSJ}$ and $Q_{J|US}$, define a new arbitrarily varying channel $\mathcal{Q}:=\{Q_{Y|UJ}:J\in\mathcal{J}\}$ such that
\begin{align}
    Q_{Y|UJ} = \sum_{s,x}Q_{S|UJ}Q_{X|US}W_{Y|XSJ},
\end{align}
where $Q_{S|UJ}=\frac{Q_{SU}Q_{J|US}}{Q_UQ_{J|U}}, Q_{J|U}=\sum_sQ_{J|US}Q_{S|U}.$
When the random state sequence $S^n$ is available at the encoder noncausally, we have the following result.
\begin{theorem}\label{the: noncausal maximal error}
    For deterministic coding under the maximal error criterion, the capacity of the message and state communication over a DMC with random states noncausally available at the encoder and arbitrary jamming is lower bounded by
    \begin{align}\label{neq: noncausal maximal lower bound}
        C_{m,n}(D)\geq \max\left\{\min_{Q_{J|U}\in\widetilde{\mathcal{P}}(\mathcal{J},Q_{S|U})}\{I(U;Y),D(Q_U)\} -I(U;S)\right\},
    \end{align}
    where
    \begin{align}
        &\widetilde{\mathcal{P}}(\mathcal{J},Q_{S|U}) = \{Q_{J|U}:  Q_{J|U}=\sum_sQ_{J|US}Q_{S|U}\;\text{for $Q_{J|US}\in\mathcal{P}(\mathcal{J})$}\}\\
        &D(Q_U) = \min_{\substack{Q_{U\widetilde{U}}\;s.t.:\\Q_U=Q_{\widetilde{U}},\\Pr\{U\overset{\mathcal{Q}}{\sim} \widetilde{U}\}=1}} I(U;\widetilde{U}),\\
        &Q_U = \sum_s Q_{U|S}Q_S,
    \end{align}
    and the maximum runs over all distributions $Q_{U|S},Q_{X|US}$ and deterministic functions $h:\mathcal{U}\times\mathcal{Y}\to\hat{\mathcal{S}}$ such that the graph of $Q_{Y|UJ}$ contains isolated vertices and $\max_{Q_J}\mathbb{E}[d(S,h(U,Y))]\leq D$.
\end{theorem}
\begin{remark}
    Here we discuss some degenerate cases.
    \begin{itemize}
        \item When $\mathcal{S}=1$, by setting $U=X$, the result reduces to the lower bound of the maximal error capacity of AVCs \cite{csiszar1981capacity}.
        \item When $\mathcal{J}=1$, the problem reduces to joint message transmission and state reconstruction with noncausal observation at the encoder, and our result reduces to \cite[Theorem 1]{choudhuri2012non}.
        \item When $\mathcal{J}=1$ and there exists a deterministic function $h$ such that $D\geq \max_{Q_{UX|S}}\mathbb{E}[d(S,h(U,Y))],$ the distortion constraint is inactive and the result reduces to the Gel'fand-Pinsker Theorem\cite[Section 7.6]{el2011network}.
    \end{itemize}
\end{remark}
The proof is provided in Appendix \ref{sec: proof of achievability of noncausal case}.
The coding scheme for Theorem \ref{the: noncausal maximal error} is mainly based on the scheme proposed in \cite{csiszar1981capacity} for arbitrarily varying channels under the maximal error criterion. The coding scheme in \cite{csiszar1981capacity} first constructs a codebook by random selection and then chooses the codewords with a specific property (a deterministic selection). Specifically, it defines a pair of codewords $(x^n_i,x^n_j)$ as bad codewords if the mutual information computed based on their joint empirical distribution is higher than the rate of the codebook. One can consider such a pair of codewords highly correlated and, hence, hard to distinguish. These bad codewords are all deleted after the random selection in \cite{csiszar1981capacity}. Such a codebook construction is not suitable for our case. This is because when noncausal information is available at the encoder, the encoder can use the Gel'fand-Pinsker coding to achieve a higher communication rate, and the evaluation of the coding scheme's performance relies on a random coding argument, while choosing codewords with a specific property makes the codebook deterministic. Instead, we enhance the result in \cite{csiszar1981capacity} by proving that most of the codewords in a randomly generated codebook have the desired property, which are called `good' codewords, and leaving those bad codewords in the codebook is negligible when evaluating the decoding error probability. 

When considering the maximal error criterion, the distribution of the message can be arbitrary. However, an argument similar to \eqref{neq: probability within the bin} implies that given each message $m\in\mathcal{M}$ and taking the average over all random states $S^n$, the codewords $U^n$ in subcodebook $\mathcal{C}(M)$ are selected almost uniformly at random. To ensure the probability of a bad codeword being selected is small, we need to uniformly distribute these codewords to all the subcodebooks. To this end, we use a random binning operation on the codebook to construct subcodebooks instead of the deterministic binning as we usually do for the Gel'fand-Pinsker coding.

\begin{corollary}\label{coro: strictly causal maximal error}
    If the graph of the channel $Q_{Y|XJ}$ has isolated vertices, the capacity--distortion function of the strictly causal case using deterministic coding under the maximal error criterion has the following lower bound:
    \begin{align}\label{eq: strictly causal maximal error lower bound}
        C_{m,s}(D) \geq \max\left\{\min\{\min_{Q_{J|X}}I(X;Y),D(Q_X)\} + \min_{Q_{J|X}} I(U;Y|X) - I(U;S|X)\right\},
    \end{align}
    where $D(Q_X):=\min_{Pr\{X \overset{\mathcal{Q}_{Y|XJ}}{\sim} \widetilde{X}\}=1} I(X;\widetilde{X}),Q_{Y|XJ}=\sum_s W_{Y|XSJ}Q_S$, the maximization runs over all distributions $Q_{XU}$ and deterministic functions $h:\mathcal{X}\times\mathcal{U}\times\mathcal{Y}\to\hat{\mathcal{S}}$ such that $\mathbb{I}_{\mathcal{U}|\mathcal{X}}(Q_{Y|UXJ})=1$ and $\max_{Q_J}\mathbb{E}[d(S,h(X,U,Y))]\leq D$.
\end{corollary}
We still use the block Markov coding scheme as of Theorem \ref{the: achievability of strictly causal} and decode $X^n$ under the maximal error criterion since it carries the message whose distribution can be arbitrary. Once the decoding of $X^n$ is correct, the decoder has the bin index of $U^n$, and the decoding of $U^n$ is almost the same as that of Theorem $\ref{the: achievability of strictly causal}$. The proof is given in Appendix \ref{sec: proof of corollary strictly causal maximal error}.

When the graph of the AVC $Q_{Y|XJ}$ is a complete graph, the value of \eqref{eq: strictly causal maximal error lower bound} is zero since in this case, any pair of $(x,\widetilde{x})$ are connected, which implies any joint distribution $Q_{X\widetilde{X}}$ satisfies $Pr\{X\overset{\mathcal{Q}}{\sim}\widetilde{X}\}=1$. One can choose the joint distribution such that $Q_{X\widetilde{X}}=Q_XQ_{\widetilde{X}}$ and $D(Q_X)=0$. The optimal input distribution in this case is to set $U=\emptyset$ and the value of \eqref{eq: strictly causal maximal error lower bound} is 0.

\section{binary case}\label{sec: binary case}

In this section, we consider channels with a binary output alphabet. Although the capacity of general AVCs under the maximal error criterion is still an open problem, AVCs with binary output have been completely solved \cite{ahlswede1970capacity}. The proof starts with the binary-input case and then extends it to the non-binary case. The graph of a binary-input binary-output AVC is either a complete graph or contains only isolated vertices. Therefore, its capacity can be completely characterized. In addition, for the binary-input case, the zero-capacity conditions for average and maximal criteria are the same by setting $Q_1(j) = T_{J|X}(j | x_2)$ and $Q_2(j) = T_{J|X_2}(j|x_1)$. However, such an equivalence does not hold for non-binary input cases.

\subsection{Noncausal Case}

\begin{corollary}\label{coro: noncausal binary output lower bound}
    When the output alphabet $|\mathcal{Y}|=2$, a deterministic coding lower bound of the maximal error capacity--distortion function $C_{m,n}(D)$ is
    \begin{align}
        C_{m,n}(D) \geq \max \left\{\min_{Q_{J|U}}I(U;Y) - I(U;S)\right\},
    \end{align}
    where the maximum is taken over all $Q_{U|S},Q_{X|US}$ and deterministic functions $h:\mathcal{U}\times\mathcal{Y}\to\hat{\mathcal{S}}$ such that $\max_{Q_{J|U}}\mathbb{E}[d(S,h(U,Y))]\leq D$.
\end{corollary}
It should be noted that for binary-output cases, there is no need to discuss the graph of the channel $\mathcal{W}$. For the binary-input case, if the graph is complete, by the definition of \eqref{def: maximal error symmetrizability}, one can set $P_{J|X}(j|x_1) = Q_1(j)$ and $P_{J|X}(j|x_2) = Q_2(j)$, which means the row-convex extension $\overline{\overline{W}}$ contains a channel with two same rows, and results in the zero capacity. The argument can be extended to non-binary input cases\cite{csiszar2011information}.

\subsection{Strictly Causal Case}
For the strictly causal case, when the state observation has a delay $d=n$, Corollary \ref{coro: strictly causal maximal error} becomes a tight bound.
\begin{corollary}
    For a binary output channel, if it is nonsymmetrizable-$\mathcal{S}|\mathcal{X}$, the deterministic coding maximal error capacity of the joint message and lossless state communication is
    \begin{align}
        C_{m,s}=\max\left\{\min_{Q_{J|X},Q_{J|X}'} [I_{Q_{J|X}}(X;Y) -  H_{Q_{J|X}'}(S|X,Y)]^+\right\}
    \end{align}
    where the maximum runs over all input distributions $P_X$ and function $h$ such that $\max_{Q_{J|X}}\mathbb{E}[d(S,h(X,U,Y))]\leq D$.
\end{corollary}
\begin{proof}
    The achievability follows from a similar argument about binary-output AVCs in Corollary \ref{coro: noncausal binary output lower bound} and setting $U=S$ in Corollary \ref{coro: strictly causal maximal error}. To prove the converse, we first use \cite[Corollary 12.3]{csiszar2011information}, which claims that the capacity of an AVC under the maximal criterion is upper bounded by that of the worst channel in its row-convex extension. Then, a compound channel argument similar to the converse of Corollary \ref{coro: strictly causal lossless average error} completes the proof.
\end{proof}

\begin{example}
    The following example shows that the deterministic coding capacity of the joint message and state communication under maximal error can be strictly smaller than that under the average error. Consider the channel such that $\mathcal{X}=\mathcal{Y}=\mathcal{S}=\mathcal{J}=\{0,1\}$ such that the transition matrices for each pair of $(s,j)$ is as follows:
    \begin{align}
        \begin{bmatrix}
        1 & 0 \\
        0.15 & 0.85
    \end{bmatrix}_{s=0,j=0}\;\;
    \begin{bmatrix}
        1 & 0 \\
        0.65 & 0.35
    \end{bmatrix}_{s=1,j=0}\;\;
    \begin{bmatrix}
        0.85 & 0.15 \\
        0 & 1
    \end{bmatrix}_{s=0,j=1}\;\;
    \begin{bmatrix}
        0.35 & 0.65 \\
        0 & 1
    \end{bmatrix}_{s=1,j=1}
    \end{align}
    where $Q_S(0)=0.9, Q_S(1)=0.1$. We have the following AVC after taking average over $S$:
    \begin{align}
        \begin{bmatrix}
        1 & 0 \\
        0.2 & 0.8
    \end{bmatrix}_{j=0}\;\;
    \begin{bmatrix}
        0.8 & 0.2 \\
        0 & 1
    \end{bmatrix}_{j=1}.
    \end{align}
    The resulting channel is nonsymmetrizable, and its graph is not a complete graph. Since the lossless reconstruction cost is the same for both cases, it is sufficient to show that the capacity for the average case can be positive, and the mutual information $I(X;Y)$ in the average case can be strictly greater than that in the maximal error case. Note that $I(X;Y)$ is concave in $Q_J$ and achieves the minimum when $Pr\{J=0\}=Pr\{J=1\}=\frac{1}{2}$. In this case, the channel becomes a binary symmetric channel with cross-over probability $0.1$. On the other hand, the lossless reconstruction cost is
    \begin{align}
        H(S|X,Y) = H(S) - I(S;Y|X)  \leq H(S) = h_b(0.1) \approx 0.47,
    \end{align}
    where $h_b(\cdot)$ is the binary entropy function.
    By choosing the input distribution $P_X(0)=P_X(1)=0.5$, we have
    \begin{align}
        I_{Q_{J}}(X;Y) - H_{Q_J}(S|X,Y) \geq 1 - h_b(0.1) - h_b(0.1) >0.
    \end{align}
    For the maximal error case, by choosing distribution $P_{J|X}(0|0)=0,P_{J|X}(0|1)=1$, the resulting channel is a binary symmetric channel with cross-over probability 0.2. With uniform input distribution,
    \begin{align}
        I_{Q_{J|X}}(X;Y) = 1 - h_b(0.2) < 1 - h_b(0.1).
    \end{align}
    The proof is completed.
\end{example}

\section{conclusion}\label{sec: conclusion}
Joint message and state communication under arbitrarily varying jamming is investigated in this paper. In addition to knowing the communication protocol, we consider jammers with or without being aware of the transmitted messages. The former case allows the jammer to choose jamming strategies more specifically and hence requires more robustness in the coding procedure. These two cases are modelled as communication with jamming under maximal and average error criteria. We also consider different encoder observations of the channel states under different error criteria. Inner bounds for general cases and tight bounds for special cases are provided, demonstrating the impact of the   jammer's knowledge on the performance of the systems.

The transmission of the state information significantly changes our coding scheme. For the strictly causal case, due to the lack of knowledge of the state information at the jammer, we can always use an average error code to transmit the state information no matter what error criterion is imposed on the message transmission. On the other hand, in the noncausal case with the maximal error criterion, we leave those `bad' codewords defined in \cite{csiszar1981capacity} to keep the randomness of the codebook for the state transmission. We also verified that for joint message and state communication, the capacity under maximal error could still be strictly smaller than that under the average error.

It should also be noted that throughout the paper, we assume the state sequence is unknown to the jammer, which allows us to construct coding schemes and analyse their performance by averaging over all possible state sequences. When the jammer is aware of the state sequence, for the strictly causal case, part of the indices of $X^n$ (indexed by the transmitted message and the bin index of the state description of the state sequence from the previous block) is also known to the jammer. Although a maximal error code always works in this case, it is still not clear to us if it is possible to further improve the coding rate with the average error criterion.
\ifCLASSOPTIONcaptionsoff
  \newpage
\fi

\appendices
\section{Properties of types}
In this section, we provide some properties about types and typical sequences that will be useful in the following sections. The proof can be found in \cite{csiszar1988capacity}\cite{csiszar2011information}.

\begin{applemma}[Type Counting Lemma {\cite[Lemma 2.2]{csiszar2011information}}]\label{lem: type counting lemma}
    The number of different types of sequences in $\mathcal{X}^n$ is less than $(n+1)^{\mathcal{X}}$.
\end{applemma}

\begin{applemma}[{\cite[Corollary 17.9A]{csiszar2011information}}]\label{lem: joint typical lemma}
    Consider $N=2^{nR}$ sequences $x^n(i)$ independently drawn from the distribution $P^n_X$. If $I(X;Y)<R$, to any $\tau>0$ there exists $\zeta>0$ such that
    \begin{align}
        |\frac{1}{n}\log |\{i:x^n(i)\in \mathcal{T}^n_{Q_{XY},\zeta}[y^n]\}|-(R-(X;Y))| <\tau \;\;\text{double exponentially surely}
    \end{align}
    simultaneously for all $y^n\in\mathcal{Y}^n$ with $\mathcal{T}^n_{Q_{XY},\zeta}[y^n]\neq \emptyset$.
\end{applemma}

\begin{applemma}[{\cite[Fact 1]{csiszar1988capacity}}]\label{lem: type fact 1}
    If $\mathcal{T}^n_{P_X}\neq \emptyset$,
    \begin{align}
       (n+1)^{-|\mathcal{X}|}2^{nH(X)} \leq |\mathcal{T}^n_{P_X}| \leq 2^{nH(X)}.
    \end{align}
\end{applemma}

\begin{applemma}[{\cite[Fact 3]{csiszar1988capacity}}]\label{lem: type fact 3}
    For any channel $W:\mathcal{X}\to\mathcal{Y},$
    \begin{align}
        \sum_{y^n\in\mathcal{T}_{P_{XY}}[x^n]}W^n(y^n|x^n) \leq 2^{-n(D(P_{XY}||P_X\cdot W_{Y|X})},
    \end{align}
    where $P_X$ is the type of $x^n,$ $P_{XY}$ is the joint type of $(x^n,y^n)$ and $P_X\cdot W_{Y|X}$ is a joint distribution on $\mathcal{X}\times\mathcal{Y}.$
\end{applemma}

\begin{applemma}[{\cite[Lemma 3]{csiszar1988capacity}}]\label{lem: type lemma 3}
    For any $\epsilon>0,n\geq n_0(\epsilon), N \geq 2^{n\epsilon}$ and type $P_X$, there exists codewords $x^n(1),...,x^n(N)$ in $\mathcal{X}^n$, each of type $P_X$, such that for every $x^n\in\mathcal{X}^n,j^n\in\mathcal{J}^n$ and every joint type $P_{XX'J}$, upon setting $R=\frac{1}{n}\log N$, we have
    \begin{itemize}
        \item $\left|\{l: (x^n,x^n(l),j^n)\in\mathcal{T}^n_{P_{XX'J}}\}\right|\leq \exp\left\{n(\left|R-I(X';X,S)\right|^++\epsilon)\right\}$;
        \item $\frac{1}{N}\left|i: (x^n(i),j^n)\in\mathcal{T}^n_{P_{XJ}} \right|\leq \exp\{-n\epsilon/2\}$ if $I(X;J)>\epsilon$;
        \item $\frac{1}{N}\left| i:(x^n(i),x^n(l),j^n)\in\mathcal{T}^n_{P_{XX'J}} \;\text{for some $l\neq i$} \right|\leq \exp\{-n\epsilon/2\}$, if $I(X;X',J)-\left| R - I(X';J) \right|^+ >\epsilon$.
    \end{itemize}
\end{applemma}

\begin{applemma}[Uniform Continuity of Entropy {\cite[Lemma 2.7]{csiszar2011information}}]\label{lem: uniform continuity of entropy}
    Let $P$ and $Q$ be two distributions on $\mathcal{X}$. If $\sum_x |P(x)-Q(x)|=\theta\leq1/2$, then
    \begin{align}
        |H(P) - H(Q)|\leq -\theta\log \frac{\theta}{|\mathcal{X}|}.
    \end{align}
\end{applemma}

\section{proof of lemma \ref{lem: nonsymmetrizable u|x inequality} and lemma \ref{lem: necessity of complete graph}}

In this section, we provide proofs of Lemma \ref{lem: nonsymmetrizable u|x inequality} and Lemma \ref{lem: necessity of complete graph}.

\subsection{Proof of Lemma \ref{lem: nonsymmetrizable u|x inequality}}\label{sec: proof of lem: nonsymmetrizable u|x inequality}

The proof is almost the same as \cite{csiszar1988capacity}\cite{gubner1990deterministic}\cite{hof2006deterministic}. Define
    \begin{align*}
        F(T) &= \max_{x,s,s',y}|\sum_{j}T(j|x,s)Q_{Y|XSJ}(y|x,s',j)\\
        &\quad\quad\quad\quad -\sum_{j}T(j|x,s')Q_{Y|XSJ}(y|x,s,j)|.
    \end{align*}
    for $T:\mathcal{X}\times\mathcal{S}\to\mathcal{J}.$ We always have
    \begin{align*}
        &\max_{x,s,s',y}|\sum_{j}T_1(j|x,s)Q_{Y|XSJ}(y|x,s',j)\\
        &\quad\quad\quad\quad -\sum_{j}T_2(j|x,s')Q_{Y|XSJ}(y|x,s,j)|\geq F(T),
    \end{align*}
    where $T=\frac{T_1+T_2}{2}$.
    Since the channel is nonsymmetriable-$\mathcal{S}|\mathcal{X}$ we have $F(T)>0$. Now, suppose $F(T)$ attains its minimum at $T^*$. The proof is completed by setting $\zeta = F(T^*).$

\subsection{Proof of Lemma \ref{lem: necessity of complete graph}}\label{sec: proof of lem: necessity of complete graph}
 Suppose there exists a pair of encoding/decoding strategies such that for an arbitrary pair of messages $(m,m')$, $\max_{j^n}e(m,j^n)\leq \epsilon\leq \frac{1}{2},\max_{j^n}e(m',j^n)\leq \epsilon\leq \frac{1}{2},$
    \begin{align}
        \epsilon &\geq\mathbb{E}[e(m,J^n)]\\
        &=\sum_{y^n\notin g^{-1}(m)}\mathbb{E}[W^n(y^n|x^n(m),J^n)]\\
        &=\sum_{y^n\notin g^{-1}(m)}\prod_{i=1}^n \mathbb{E}[W(y_i|x_i(m),J_i)]\\
        &\overset{(a)}{=}\sum_{y^n\notin g^{-1}(m)}\prod_{i=1}^n \sum_j W(y_i|x_i(m),j)Q_{1i}(j)\\
        &\overset{(b)}{=} \sum_{y^n\notin g^{-1}(m)}\prod_{i=1}^n \sum_j W(y_i|x_i(m'),j)Q_{2i}(j)\\
        &=\sum_{y^n\notin g^{-1}(m)}  \mathbb{E}[W^n(y^n|x^n(m'),J^n)]\\
        &\geq \sum_{y^n\in g^{-1}(m')}  \mathbb{E}[W^n(y^n|x^n(m'),J^n)] \geq 1 - \epsilon,
    \end{align}
    where $(a)$ and $(b)$ follow from the fact that the jammer can choose the jamming strategy freely, and the channel is complete-$\mathcal{X}$, so equation \eqref{def: maximal error symmetrizability} holds.
    It then follows that $\epsilon \geq \frac{1}{2}$, which is a contradiction.

\section{Proof of theorem \ref{the: achievability of strictly causal}}\label{sec: proof of achievability of strictly causal}
In this section, we prove Theorem \ref{the: achievability of strictly causal}. The proposed coding scheme is an extension of that in \cite{csiszar1988capacity}, in which we deal with arbitrarily varying channels with an additional random state $S$. To apply the analysis in \cite{csiszar1988capacity} to the codeword $X$, we first show that an average error codebook is sufficient to decode $X$ for the given channel model. Furthermore, the error analysis in \cite{csiszar1988capacity} does not apply to the decoding of $U$ in Theorem \ref{the: achievability of strictly causal} as the lossy description is determined by a channel state sequence $S^n$ and does not produce the channel output directly. To solve this problem, we consider the channel state $S^n$ as a second jamming sequence with type constraint due to the fact that it is generated i.i.d. according to the given distribution $Q_S$. With this consideration, we can build a connection between $U^n$ and $Y^n$ through $(X^n,S^n,J^n)$ by considering the induced channel $Q_{Y|XUJ}$ and bound the decoding error probability. A key step is to show that a property similar to the conditional typical lemma\cite[Section 2.5]{el2011network} still holds for $(U^n,X^n,S^n,Y^n)$ regardless of the type of the jamming sequence using the method of types.


Fix a joint distribution $Q_{YXSU|J}=Q_XQ_SQ_{U|XS}W_{Y|XSJ}.$ For given $\eta \geq 0$, define sets of joint distributions $P_{YXSJ}$ by
\begin{align}
&\Psi^1_{\eta}=\{P_{YXJ}: D(P_{YXJ} || Q_XP_JQ_{Y|XJ})\leq \eta\},\\
&\Psi^2_{\eta}=\{P_{XUS}: D(P_{XUS} || Q_XQ_SQ_{U|XS})\leq \eta\},\\
&\Psi^3_{\eta}=\{P_{YXUJ}: D(P_{YXJU} || Q_{YU|XJ}Q_{X}P_J)\leq \eta\},\\
&\Psi^4_{\eta}=\{P_{XSJU}: D(P_{XSJU} || Q_{XUS}P_{J})\leq \eta\}.
\end{align}

Set
\begin{align}
    &\widetilde{R}_S = I(U;S|X)+2\tau,\\
    &R+R_S = I(X;Y) - \tau\\
    &R_S':=\widetilde{R}_S-R_S \leq I(U;Y|X)-\tau,
\end{align}

\emph{Codebook Generation:} For each block $c$, the encoder picks $2^{n(R+R_S)}$ codewords $X^n$ uniformly at random from a fixed type set $\mathcal{T}^n_{Q_X}$. Each codeword is uniquely indexed by a pair of integers $x^n(m_c,l_c),$ where $m_c\in[1:2^{nR}],l_c\in[1:2^{nR_S}]$. For each pair $(m_c,l_c)$ and the marginal distribution $Q_{U|X}=\sum_s Q_SQ_{U|XS}$, select $2^{n\widetilde{R}_S}$ sequences $\{u^n|m_c,l_c\}$ from the fixed conditonal type set $\mathcal{T}^n_{Q_{U|X}}[x^n(m_c,l_c)]$ uniformly at random, which are lossy descriptions of the channel state sequence $S^n_c$. For each lossy description $u^n$, assign it an index uniformly at random from the integer set $[1:2^{nR_S}]$. Sequences with the same index form a bin. Denote the bin number for a given $u^n$ by $b(u^n)$. Now each lossy description $u^n$ can be uniquely indexed by $u^n(b(u^n),k_c|m_c,l_c)$, where $k_c$ is the index of $u^n$ within bin $b(u^n)$. 

With this random selection and random binning, we have a set of sub-codebooks $\{u^n|m_c,l_c,b\},b\in[1:2^{nR_S}]$, each with size $2^{nR_S'},R_S':=\widetilde{R}_S-R_S.$

\emph{Encoding.} To encode the message and state, we use the union of the sub-codebooks as the codebook. At the beginning of each block $c$, the encoder observes the channel state from the previous block $c-1$, denoted by $S^n_{c-1}.$ It looks for a lossy description $u^n\in\{u^n|m_{c-1},l_{c-1}\}$ such that
\begin{align}
    P_{u^n,x^n(m_{c-1},l_{c-1}),s^n_{c-1}}\in \Psi^2_{\eta}.
\end{align}
If there are multiple such $u^n$, choose one uniformly at random. If there is no such $u^n$, the encoder declares an error. Set $l_c = b(u^n)$. To transmit the message $m_c$, the encoder chooses codeword $x^n(m_c,l_c).$

\emph{Decoding and reconstruction.} At the end of block $c$, the decoder observes the channel output $y^n_c$. We define the decoding rule for the decoder $\phi$ as follows.

Upon observing the channel output $y^n_c$, the decoder outputs $\phi(y^n_c)=(\hat{m}_c,\hat{l}_c,\hat{k}_{c-1})$ if and only if there exists a $j^n\in\mathcal{J}^n$ such that
\begin{itemize}
    \item The joint type $P_{x^n(\hat{m}_c,\hat{l}_c),y^n_c,j^n}\in \Psi^1_{\eta}$.
    \item For each pair of competitor $(\hat{m}_c',\hat{l}_c')\neq (\hat{m}_c,\hat{l}_c)$ such that $P_{x^n(\hat{m}_c',\hat{l}_c'),y^n_c,{j^n}'}\in\Psi^1_{\eta}$ for some ${j^n}'\in\mathcal{J}^n$, we have $I(X,Y;X'|J)\leq \eta$, where $X,X',J,Y$ are dummy random variables such that the joint distribution of $(x^n(\hat{m}_c,\hat{l}_c),x^n(\hat{m}_c',\hat{l}_c'),y^n_c,j^n)$ equals $P_{XX'YJ}$.
    \item There exists an $s^n_{c-1}\in\mathcal{S}^{n}$ such that the joint type
    \begin{align}
        P_{x^n(m_{c-1},l_{c-1}),u^n(\hat{l}_c,\hat{k}_{c-1}|m_{c-1},l_{c-1}),y^n_{c-1},j^n_{c-1}}\in\Psi^3_{\eta},\\
        P_{x^n(m_{c-1},l_{c-1}),u^n(\hat{l}_c,\hat{k}_{c-1}|m_{c-1},l_{c-1}),s^n_{c-1},j^n_{c-1}}\in\Psi^4_{\eta}.
    \end{align}
    \item For each competitor $\hat{k}_{c-1}'\neq \hat{k}_{c-1}$ such that 
    \begin{align}
        &P_{x^n(m_{c-1},l_{c-1}),u^n(\hat{l}_c,\hat{k}_{c-1}'|m_{c-1},l_{c-1}),y^n_{c-1},{j^n_{c-1}}'}\in \Psi_{\eta}^3,\\
        &P_{x^n(m_{c-1},l_{c-1}),u^n(\hat{l}_c,\hat{k}_{c-1}'|m_{c-1},l_{c-1}),{s^n_{c-1}}',{j^n_{c-1}}'}\in \Psi_{\eta}^4
    \end{align}
    for some ${s^n_{c-1}}'\in\mathcal{S}^n,{j^n_{c-1}}'\in\mathcal{J}^n$, we have $I(Y,S,U;U'|J,X)\leq \eta$, where $X,U,U',Y,J,S$ are dummy random variables whose joint distribution $P_{XUU'YSJ}$ is the joint type of $x^n(m_{c-1},l_{c-1}),u^n(\hat{l}_c,\hat{k}_{c-1}|m_{c-1},l_{c-1}),u^n(\hat{l}_c,\hat{k}_{c-1}'|m_{c-1},l_{c-1}),y^n_{c-1},$ $s^n_{c-1},{j^n_{c-1}}$.

    The decoder reconstructs the sequence $\hat{s}^n$ by $\hat{s}_i=h(x_i,u_i,y_i),i=1,...,n.$
    
\end{itemize}
 The first two steps of the decoding rule are to decode the codeword $X^n_c$ for the current block $c$. Since the encoder and decoder are both not aware of the state sequence $S^n_c$ in the current block, they consider the average channel $Q_{Y|XJ}=\sum_s Q_SW_{Y|XSJ}$. The third and fourth steps are to find the lossy description $U^n$ that describes the state sequence $S^n_{c-1}$ from the previous block. To find the description that matches the state sequence, the decoder considers the state sequence $S^n_{c-1}$ as a second jamming sequence. However, since the distribution of $S$, the encoding strategy and the way that $S^n$ is generated are all known to all the participants of the system, this second jamming sequence has additional type constraints defined in $\Psi^4_{\eta}$. With this additional constraint, the joint empirical distribution of $(U^n,Y^n_{c-1},S^n_{c-1})$ is close to their underlying distribution, which ensures the distortion constraint.

In the next subsection, we first show that the decoder we defined above is unambiguous, which implies that there do not exist two codeword pairs $(X^n,U^n)$ and $({X^n}',{U^n}')$ that satisfy the decoding rule simultaneously. This indicates that the decoder always outputs a unique codeword pair based on the decoding rules and it remains to bound the probability that the decoder outputs a wrong codeword pair.

\subsection{Unambiguity Analysis}
In this subsection, we show that the decoding rule defined in the last subsection is unambiguous. 

The unambiguity of the decoding rule is ensured by the following lemma.
\begin{lemma}
    For the given channel $W_{Y|XSJ}$, if the corresponding channel $Q_{Y|XUJ}$ is nonsymmetrizable-$\mathcal{U}| \mathcal{X}$ and $Q_{Y|XJ}$ is nonsymmetrizable-$\mathcal{X}$ and $\beta_1>0,\beta_2>0$, then for a sufficiently small $\eta$, no tuple of random variables $(X_1,X_2,X_2',U_1,U_1',S_1,S_1',$ $J_1,J_1',J_2,J_2',Y_1,Y_2)$ can simultaneously satisfy
    \begin{align}
        \label{eq: average X 1}&P_{X_1}=P_{X_2}=P_{X_2'}=Q_X, \min_{x}Q_X(x) \geq \beta_1,\\
        &P_{U_1|X_1}=P_{U_1'|X_1}=Q_{U|X}, \min_{u,x}Q_{U|X}(u|x) \geq \beta_2,\\
        \label{eq: average X 2}&P_{X_2J_2Y_2}\in \Psi^1_{\eta},\;\; P_{X_2'J_2'Y_2}\in \Psi^1_{\eta},\\
        \label{eq: average X 3}&I(X,Y;X'|J)\leq \eta, I(X',Y;X|J') \leq \eta,\\
        &P_{X_1U_1J_1Y_1}\in \Psi^3_{\eta},\;\; P_{X_1U_1'J_1'Y_1}\in \Psi^3_{\eta},\\
        &P_{X_1U_1S_1J_1}\in \Psi^4_{\eta},\;\; P_{X_1U_1'S_1'J_1'}\in \Psi^4_{\eta},\\
        &I(U_1,Y_1,S_1;U_1'|J_1,X_1)\leq \eta,I(Y_1,U_1',S_1';U_1|J_1',X_1)\leq \eta.
    \end{align} 
    The subscripts $1$ and $2$ are to indicate the order of the blocks from which those random variables come.
\end{lemma}
 As we defined in \emph{Decoding}, random variables with a prime are those defined by the type of some wrong codewords. Conditions \eqref{eq: average X 1}, \eqref{eq: average X 2} and \eqref{eq: average X 3} from \cite{csiszar1988capacity} ensure the unambiguity of the codeword $X^n$. The remaining conditions ensure the unambiguity of $U^n$. They claim that one cannot find two pairs of $(U^n,X^n)$ and $({U^n}',{X^n}')$ that match the output sequence $Y^n$, together with the additional type constraint in $\Psi^4_{\eta}.$

The proofs of the conditions \eqref{eq: average X 1},\eqref{eq: average X 2} and \eqref{eq: average X 3} are provided in \cite{csiszar1988capacity}. Hence, we only analyze the unambiguous requirements of the remaining conditions. By the definition of $\Psi_{\eta}^3,$
\begin{align}
    \label{uman: kl}&D(P_{YXUJ}||Q_XP_JQ_{YU|XJ})\\
    &=\sum_{x,u,y,j}P_{YXUJ}(y,x,u,j)\log \frac{P_{YXUJ}(y,x,u,j)}{Q_X(x)P_J(j)Q_{YU|XJ}(y,u|x,j)}\\
    &=\sum_{x,u,y,j}P_{YXUJ}(y,x,u,j)\log \frac{P_{YXUJ}(y,x,u,j)}{Q_X(x)P_J(j)\sum_s Q_S(s) Q_{U|XS}(u|x,s)W_{Y|XSJ}(y|x,s,j)}\leq \eta,
\end{align}
where the last equality is by the definition $Q_{YU|XJ}(y,u|x,j)=\sum_s Q_S(s) Q_{U|XS}(u|x,s)W_{Y|XSJ}(y|x,s,j)$. By the definition of mutual information,
\begin{align}
    \label{uman: mu}I(U,Y;U'|J,X) = \sum_{x,u,u',j,y}P_{XUU'JY}(x,u,u',j,y)\log \frac{P_{U'|JXUY}(u'|j,x,u,y)}{P_{U'|JX}(u'|j,x)}\leq I(U,Y,S;U'|J,X) \leq \eta.
\end{align}
Note that it is without loss of generality to assume $P_{XJ}(x,j)>0$, otherwise the corresponding terms in \eqref{uman: kl} and \eqref{uman: mu} are zero and has no impact on the analysis.
Adding \eqref{uman: mu} to \eqref{uman: kl} gives
\begin{align}
    2\eta \geq \sum_{x,u,u',j,y}P_{XUU'JY}(x,u,u',j,y)\log \frac{P_{JXUU'Y}(j,x,u,u',y)}{Q_X(x)P_J(j)P_{U'|JX}(u'|j,x)\sum_s Q_S(s) Q_{U|XS}(u|x,s)W_{Y|XSJ}(y|x,s,j)}.
\end{align}
Projecting the above inequality onto $\mathcal{X}\times\mathcal{U}\times\mathcal{U}'\times\mathcal{Y}$ yields
\begin{align}
    &\sum_{x,u,u',y}P_{XUU'Y}(x,u,u',y)\log \frac{P_{YXUU'}(y,x,u,u')}{Q_X(x)\sum_jP_J(j)P_{U'|JX}(u'|j,x)\sum_s Q_S(s) Q_{U|XS}(u|x,s)W_{Y|XSJ}(y|x,s,j)}\\
    &=D(P_{XUU'Y}||Q_X\sum_j P_JP_{U'|JX}Q_{YU|XJ})\leq 2\eta.
\end{align}
By Pinsker's inequality, it follows that
\begin{align}
    &\sum_{x,u,u',y}|P_{XUU'Y}(x,u,u',y) - Q_X(x)\sum_{j}P_J(j)P_{U'|JX}(u'|j,x)\sum_{s}Q_S(s)Q_{U|XS}(u|x,s)W_{Y|XSJ}(y|x,s,j)| \notag\\
    &\overset{(a)}{=}\sum_{x,u,u',y}|P_{XUU'Y}(x,u,u',y) - Q_X(x)\sum_{j}P_J(j)\frac{P_{U'JX}(u',j,x)}{P_{JX}(j,x)}\sum_{s}Q_S(s)Q_{U|XS}(u|x,s)W_{Y|XSJ}(y|x,s,j)|\notag\\
    &=\sum_{x,u,u',y}|P_{XUU'Y}(x,u,u',y) - Q_X(x)\sum_{j}P_J(j)\frac{P_{J|U'X}(j|u',x)P_{U'X}(u',x)}{P_{JX}(j,x)}\sum_{s}Q_S(s)Q_{U|XS}(u|x,s)W_{Y|XSJ}(y|x,s,j)|\notag\\
    &=\sum_{x,u,u',y}|P_{XUU'Y}(x,u,u',y) - Q_X(x)P_{U'|X}(u'|x)\sum_{j}\frac{P_J(j)Q_{X}(x)}{P_{JX}(j,x)}P_{J|U'X}(j|u',x)\sum_{s}Q_S(s)Q_{U|XS}(u|x,s)W_{Y|XSJ}(y|x,s,j)|\notag\\
    &\overset{(b)}{=}\sum_{x,u,u',y}|P_{XUU'Y}(x,u,u',y) \notag\\
    &\quad -Q_X(x)P_{U'|X}(u'|x)\sum_{j}\frac{P_J(j)Q_{X}(x)}{P_{JX}(j,x)}P_{J|U'X}(j|u',x)\sum_{s}Q_S(s)\frac{Q_X(x)P_{U|X}(u|x)Q_{S|UX}(s|u,x)}{Q_{XS}(x,s)}W_{Y|XSJ}(y|x,s,j)|\notag\\
    &\overset{(c)}{=}\sum_{x,u,u',y}|P_{XUU'Y}(x,u,u',y) \notag\\
    &\quad -Q_X(x)P_{U'|X}(u'|x)P_{U|X}(u|x)\sum_{j}\frac{P_J(j)Q_{X}(x)}{P_{JX}(j,x)}P_{J|U'X}(j|u',x)\sum_{s}\frac{Q_X(x)Q_S(s)}{Q_{X}(x)Q_S(s)}Q_{S|UX}(s|u,x)W_{Y|XSJ}(y|x,s,j)|\notag\\
    \label{ine: average error analysis 1}&\leq a\sqrt{2\eta}
\end{align}
for some constant $a>0,$ where $(a)$ and $(b)$ are by applying the Bayes rule to $P_{U'|JX}(u'|j,x)$ and $Q_{U|XS}(u|x,s)$, $(c)$ follows by the definition of the distribution $Q_{XS}(x,s)=Q_X(x)Q_S(s)$.

Similarly, for $D(P_{YXU'J'}||Q_XP_{J'}Q_{YU|XJ'})$ and $I(U',S',Y;U|J',X)$, we replace $(U,U',J)$ in the above inequalities with $(U',U,J')$, which gives
\begin{align}
    &\sum_{x,u,u',y}|P_{XUU'Y}(x,u,u',y) - Q_X(x)\sum_{j}P_{J'}(j)P_{U|J'X}(u|j,x)\sum_{s}Q_S(s)Q_{U|XS}(u'|x,s)W_{Y|XSJ'}(y|x,s,j)|\notag \\
    &=\sum_{x,u,u',y}|P_{XUU'Y}(x,u,u',y)\notag \\
    &\quad -Q_X(x)P_{U'|X}(u'|x)P_{U|X}(u|x)\sum_{j}\frac{P_{J'}(j)Q_{X}(x)}{P_{J'X}(j,x)}P_{J'|UX}(j|u,x)\sum_{s}Q_{S|UX}(s|u',x)W_{Y|XSJ}(y|x,s,j)|\notag\\
    \label{ine: average error analysis 2}&\leq a\sqrt{2\eta}.
\end{align}
Note that $Q_{S|UX}=Q_{S|U'X}$ and $W_{Y|XSJ}=W_{Y|XSJ'}$ since they are given distributions. Applying the triangle inequality to \eqref{ine: average error analysis 1} and \eqref{ine: average error analysis 2} gives
\begin{align*}
    &\sum_{x,u,u',y}Q_X(x)P_{U'|X}(u'|x)P_{U|X}(u|x)\\
    &\cdot\Bigg|\sum_{j}\frac{P_J(j)Q_{X}(x)}{P_{JX}(j,x)}P_{J|U'X}(j|u',x)\sum_{s}Q_{S|UX}(s|u,x)W_{Y|XSJ}(y|x,s,j)\\
    &-\sum_{j}\frac{P_{J'}(j)Q_{X}(x)}{P_{J'X}(j,x)}P_{J'|UX}(j|u,x)\sum_{s}Q_{S|UX}(s|u',x)W_{Y|XSJ}(y|x,s,j)\Bigg|\\
    &\leq 2a\sqrt{2\eta}.
\end{align*}
By our assumptions that $\min_{x}Q_{X}(x) \geq \beta_1,\min_{u,x}Q_{U|X}(u|x) \geq \beta_2$,
\begin{align}
    &\sum_{x,u,u',y}\Bigg|\sum_{j}\frac{P_J(j)Q_{X}(x)}{P_{JX}(j,x)}P_{J|U'X}(j|u',x)\sum_{s}Q_{S|UX}(s|u,x)W_{Y|XSJ}(y|x,s,j)\notag\\
    &-\sum_{j}\frac{P_{J'}(j)Q_{X}(x)}{P_{J'X}(j,x)}P_{J'|UX}(j|u,x)\sum_{s}Q_{S|UX}(s|u',x)W_{Y|XSJ}(y|x,s,j)\Bigg|\notag\\
    \label{neq: bound 1}&\leq \frac{2a\sqrt{2\eta}}{\beta_1\beta_2^2}.
\end{align}
Note that projecting the KL divergence $D(P_{YXUJ}||Q_XP_{J}Q_{YU|XJ})$ and $D(P_{YXU'J'}||Q_XP_{J'}Q_{YU|XJ'})$ to $\mathcal{X}\times\mathcal{J}$ gives
\begin{equation}\label{neq: distance qxpj}
    \begin{aligned}
        &\sum_{x,j}|Q_X(x)P_J(j)-P_{XJ}(x,j)|\leq \widetilde{a}\sqrt{2\eta},\\
    &\sum_{x,j}|Q_X(x)P_{J'}(j)-P_{XJ'}(x,j)|\leq \widetilde{a}\sqrt{2\eta}.
    \end{aligned}
\end{equation}
Substituting \eqref{neq: distance qxpj} back to \eqref{neq: bound 1} gives
\begin{align}
    &\sum_{x,u,u',y}\Bigg|\sum_{j}P_{J|U'X}(j|u',x)\sum_{s}Q_{S|UX}(s|u,x)W_{Y|XSJ}(y|x,s,j)\\
    &-\sum_{j}P_{J'|UX}(j|u,x)\sum_{s}Q_{S|UX}(s|u',x)W_{Y|XSJ}(y|x,s,j)\Bigg|\\
    \label{eq:  channel qyuj upper bound}&=\sum_{x,u,u',y}\Bigg|  \sum_{j}P_{J|U'X}(j|u',x)Q_{Y|XUJ}(y|x,u,j)-\sum_j P_{J'|UX}(j|u,x)Q_{Y|XUJ}(y|x,u',j) \Bigg| \leq f(\eta)
\end{align}
where $f(\eta)\to 0$ as $\eta \to 0$, and 
\begin{align}
    Q_{Y|XUJ}(y|x,u,j) = \sum_{s}Q_{S|UX}(s|u,x)W_{Y|XSJ}(y|x,s,j).
\end{align}

By Lemma \ref{lem: nonsymmetrizable u|x inequality}, if the channel $Q_{Y|XUJ}$ is nonsymmetrizable-$\mathcal{U}|\mathcal{X}$, there exists a positive real number $\zeta$ such that for each pair of distributions $(P_{J|U'X},P_{J'|UX})$,
\begin{align}
    \max_{x,u,u',y}|\sum_{j}P_{J|U'X}(j|u',x)Q_{Y|XUJ}(y|x,u,j)-\sum_{j}P_{J'|UX}(j|u,x)Q_{Y|XUJ}(y|x,u',j)|\geq \zeta.
\end{align}
However, if we choose a sufficiently small $\eta$ in \eqref{eq:  channel qyuj upper bound} such that $f(\eta)<\zeta$, there is a contradiction. This completes the proof.

\subsection{Error Analysis}
After showing the unambiguity of the decoder, it remains to prove that the probability that a wrong codeword pair $({X^n}',{U^n}')$ is the only codeword pair satisfying the decoding rule is small. To this end, for each block, define
\begin{align}
    \label{ine: overall error probability of stricly causal average error}&e_s(j^n) = \frac{1}{2^{nR}}\sum_{m=1}^{2^{nR}}\sum_{s^n} Q_S(s^n)\sum_{l}P_L(l)\sum_{u^n\in\{u^n|m,l\}}Pr\{U^n=u^n(b,k)|x^n(m,l),s^n\}\\
    &\quad\quad\quad \cdot \sum_{\substack{y^n:\phi(y^n)= (\hat{m},\hat{l},\hat{k})\\ s.t. (\hat{m},\hat{l},\hat{k})\neq (m,l,k)}} W^n_{Y|XSJ}(y^n|x^n(m,l),s^n,j^n),
\end{align}
where $P_L$ is the distribution on the second index of $X^n$ that is determined by the lossy description of the channel state sequence in the previous block. Clearly, we have
\begin{align}
    e_s(j^n) = Pr\{(\hat{X}^n,\hat{U}^n)\neq (X^n,U^n)|j^n\},
\end{align}
which is upper bounded by
\begin{align}
    Pr\{\hat{X}^n\neq X^n|j^n\} + Pr\{\hat{U}^n\neq U^n|\hat{X}^n=X^n,j^n\}.
\end{align}

At the beginning of this subsection, we first prove that the decoding error of $X^n$ is small. It is sufficient to show that one can construct a code for AVCs under the average error criterion to reliably decode $X^n$, and the remaining error analysis can be found in \cite{csiszar1988capacity}.
To obtain the error bound on the decoding of $X^n$, it is sufficient to consider
\begin{align}
    \frac{1}{2^{nR}}\sum_{m=1}^{2^{nR}}\sum_{l=1}^{2^{nR_S}}P_L(l)\sum_{s^n}Q_S^n(s^n)\sum_{y^n:\phi(y^n)\neq (m,l)} W^n_{Y|XSJ}(y^n|x^n(m,l),s^n,j^n),
\end{align}
which is also an upper bound of $e_{a,s}(j^n)$ defined in Definition \ref{def: achievability of average error}.
It remains to argue the distribution of $L$. We start with proving the following lemma.
\begin{applemma}\label{lem: usx typicality lemma}
    Let $Q_{XSU}=Q_XQ_SQ_{U|XS}$ be a fixed distribution and $Q_{U|X}$ be its marginal distribution. For any $x^n\in \mathcal{T}^n_{Q_X},s^n$ generated according to $Q_S$ and $U^n\in \mathcal{T}^n_{Q_{U|X}}[x^n]$, we have
    \begin{align}
        Pr\left\{ (x^n,s^n,U^n)\in \mathcal{T}^n_{Q_{XSU},\delta} \right\} \overset{\cdot}{=} 2^{-nI(U;S|X)}.
    \end{align}
\end{applemma}
\begin{proof}
We start the analysis by showing that for any $x^n\in\mathcal{T}^n_{Q_X}$ and $\delta>0$ such that $\delta\to 0$ as $n\to \infty$, we have
\begin{align}
    \label{def: x and s typicality}Pr\left\{\Big| N(x,s|x^n,S^n) - N(x|x^n)Q_{S|X}(s|x)   \Big| \geq n\delta \right\}\to 0.
\end{align}
To see this, it follows that
\begin{align}
    &Pr\left\{\Big| N(x,s|x^n,S^n) - N(x|x^n)Q_{S|X}(s|x)   \Big| \geq n\delta \right\} \notag\\
    &=Pr\left\{\Big| \frac{N(x,s|x^n,S^n)}{N(x|x^n)} - Q_{S}(s)   \Big| \geq \frac{n\delta}{N(x|x^n)} \right\} \notag \\
    &=Pr\left\{\Big| \frac{N(x,s|x^n,S^n)}{N(x|x^n)} - Q_{S}(s)   \Big| \geq \frac{\delta}{Q_X(x)} \right\} \notag\\
    &\leq Pr\left\{\Big| \frac{N(x,s|x^n,S^n)}{N(x|x^n)} - Q_{S}(s)   \Big| \geq \delta \right\}. \notag
\end{align}
Since each component of $S^n$ is generated according to the fixed distribution $Q_S$ identically and independently, we have 
\begin{align}
    \frac{N(x,s|x^n,S^n)}{N(x|x^n)} \to Q_{S}(s)\;\;\text{in probability.}
\end{align}
The remaining proof follows similarly to the proof of \cite[Lemma 2.13]{csiszar2011information}. The set $\mathcal{T}^n_{Q_{U|XS},\delta}$ is a union of at most $(n+1)^{|\mathcal{X}||\mathcal{S}||\mathcal{U}|}$ disjoint $V-$shell $\mathcal{T}_{V_{U|XS}}^n$ such that
\begin{align}
    |H(V|P_{x^n,s^n}) - H(Q_{U|XS}|P_{x^n,s^n})| \leq |\mathcal{X}||\mathcal{S}||\mathcal{U}|\delta'\log \delta'.
\end{align}
By \cite[Lemma 2.5]{csiszar2011information} we have
\begin{align}
    (n+1)^{-|\mathcal{X}||\mathcal{S}||\mathcal{U}|}2^{n(H(Q_{U|XS}|P_{x^n,s^n})-|\mathcal{X}||\mathcal{S}||\mathcal{U}|\delta'\log \delta')}\leq |\mathcal{T}^n_{Q_{U|XS},\delta}[x^n,s^n]|\leq 2^{n(H(Q_{U|XS}|P_{x^n,s^n})+|\mathcal{X}||\mathcal{S}||\mathcal{U}|\delta'\log \delta')}.
\end{align}
Further note that \eqref{def: x and s typicality} implies
\begin{align}
    |P_{x^n,s^n} - Q_XQ_S| \leq \delta
\end{align}
and hence
\begin{align}
    (n+1)^{-|\mathcal{X}||\mathcal{S}||\mathcal{U}|}2^{n(H(U|X,S)-o(\delta,\delta')}\leq |\mathcal{T}^n_{Q_{U|XS},\delta}[x^n,s^n]|\leq 2^{n(H(U|X,S)+o(\delta,\delta')}
\end{align}
To conclude the proof, we have
\begin{align}
    Pr\{U^n\in\mathcal{T}_{Q_{U|XS},\delta}[x^n,s^n]\}\overset{\cdot}{=} 2^{nH(U|XS)}\frac{1}{2^{nH(U|X)}} =2^{-nI(U;S|X)}.
\end{align}
This completes the proof.
\end{proof}

Note that by our coding scheme, for block $c$, $L_c$ is determined by $U^n_{c-1}$, which is the lossy description of the channel state $S^n_{c-1}$ from the previous block $c-1$. The distribution of $L_
c$ satisfies
\begin{align}
    &\sum_{s^n\in\mathcal{S}^n}\sum_{i=1}^{2^{n(R_S+R_S')}}Pr\{S^n_{c-1}=s^n,\text{$U^n_{c-1}(i)$ is selected},L_c=l_c|X^n_{c-1}=x^n\} \\
    &= \sum_{s^n\in\mathcal{S}^n}Q_S^n(s^n)\sum_{i=1}^{2^{n(R_S+R_S')}} Pr\{U^n_{c-1}(i)\;\text{is selected}|s^n,x^n\} Pr\{bin(U^n_{c-1}(i))=l_c|U^n_{c-1}(i)\;\text{is selected}\} \\
    &= \sum_{s^n\in\mathcal{S}^n}Q_S^n(s^n)\sum_{i=1}^{2^{n(R_S+R_S')}} Pr\{(U^n_{c-1}(i),s^n,x^n)\in\mathcal{T}^n_{Q_{XSU},\delta}[s^n,x^n]\}\\
    &\quad\quad\quad\quad\quad Pr\{U^n_{c-1}(i)\;\text{is selected}|(U^n_{c-1}(i),s^n,x^n)\in\mathcal{T}^n_{Q_{XSU},\delta}[s^n,x^n]\} Pr\{bin(U^n_{c-1}(i))=l_c|U^n_{c-1}(i)\;\text{is selected}\} \\
    &\overset{\cdot}{=} \sum_{s^n\in\mathcal{S}^n}Q_S^n(s^n) 2^{n(R_S+R_S')}2^{-n(I(U;S|X))}2^{-n(R_S+R_S'-I(U;S|X))}2^{-nR_S}\\
    \label{eq: probability of lc}&\overset{\cdot}{=}2^{-nR_S},
\end{align}
where the last two equations in exponential scale follow from Lemma \ref{lem: usx typicality lemma} and Lemma \ref{lem: joint typical lemma}. Hence, the reliability analysis of $X^n$ in \cite{csiszar2011information} still holds for our problem with a  sufficiently large $n$. The decoding error is arbitrarily small if
\begin{align}
    R+R_S \leq \min_{Q_J}I(X;Y).
\end{align}
In fact, it further follows from \eqref{eq: probability of lc} that for any $i\in 2^{n(R_S+R_S')}$,
\begin{align}
    &Pr\{U^n(i)\;\text{is selected}|x^n\}\\
    \label{ine: probability of LK}&\overset{\cdot}{=} 2^{-n(R_S+R_S')}=2^{-nR_{\widetilde{S}}}.
\end{align}
for $i\in 2^{n(R_S+R_S')}$ given $x^n.$ Hence, for given $x^n(m,l)$ and the bin index $b$, which determines the bin $\{u^n|m,l,b\}$, the probability that the selected lossy description is indexed within the bin by $k$ satisfies
\begin{align}
    \label{neq: probability within the bin}Pr\{U^n(m,l,b,k) \;\text{is selected}|x^n(m,l),b\} \overset{\cdot}{=} 2^{-nR_S'}.
\end{align}

The probability of each realization $s^n$ is that if $s^n \in \mathcal{T}^n_{Q_{S|UX},\delta}[x^n,u^n]$,
\begin{align}
    Pr\{S^n=s^n|x^n,u^n\} \leq 2^{-n(H(Q_{S|UX}|Q_{UX})-\delta)}
\end{align}
and $0$ otherwise. Now, suppose the codeword $X^n$ is correctly decoded. In the following, we bound the decoding error probability of $U^n$.

It should be noted that when $P_{x^n,u^n,j^n,y^n}\in\Psi^3_{\eta}$ and $P_{x^n,u^n,s^n,j^n}\in\Psi^4_{\eta}$ but the decoder outputs a wrong codeword, the last condition in the decoding rule must be violated. We define the following events for block $c$ (the decoding of $U^n_{c-1}$ is performed in block $c$):
\begin{align}
    &\mathcal{E}_0=\{(U^n,X^n_{c-1},S^n_{c-1})\notin \mathcal{T}^n_{Q_{XSU},\delta}, \;\text{$\forall U^n\in\{U^n|m_{c-1},l_{c-1}\}$}\},\\
    &\mathcal{E}_1=\{(U^n_{c-1},X^n_{c-1},S^n_{c-1},J^n_{c-1})\notin \Psi^4_{\eta}\}\\
    &\mathcal{E}_2=\{(U^n_{c-1},X^n_{c-1},J^n_{c-1},Y^n_{c-1})\notin \Psi^3_{\eta}\},\\
    &\mathcal{E}_3=\{\text{${U^n}'$ is output by the decoder for some ${U^n}'\neq U^n_{c-1}.$}\}.
\end{align}
In the following analysis, we omit the index $c-1$ for simplicity.
\subsubsection{Error event $\mathcal{E}_0$}

 For any $s^n\in\mathcal{T}^n_{Q_S,\delta}$ and $(x^n(m,l),s^n)\in \mathcal{T}^n_{Q_{XS},2\delta}$ (which is ensured by \eqref{def: x and s typicality} with high probability), let $\{u^n|m,l,s^n\}$ be the subset of $\{u^n|m,l\}$ such that  $(u^n,x^n(m,l),s^n)\in\mathcal{T}^n_{Q_{XUS},3\delta}$. 
 By Lemma \ref{lem: usx typicality lemma} and Lemma \ref{lem: joint typical lemma}, it follows that with double exponentially surely,
 \begin{align}
     \Big| \frac{1}{n}\log |\{u^n|m,l,s^n\}| - 2\tau  \Big| \leq \tau.
 \end{align}
 This proves the bound
 \begin{align}
     Pr\{\mathcal{E}_0\} \leq 2^{-n\nu}.
 \end{align}
 for some $\nu>0.$
 \subsubsection{Error event $\mathcal{E}_1$}
For fixed $x^n,j^n$ and $u^n$, define the set of the channel states
\begin{align}
    \mathcal{S}^n_{\epsilon}(x^n,u^n,j^n) = \{s^n \;\text{s.t. $D(P_{SUX}||Q_{SUX})\leq \eta$}: I(P_{s^n|x^n,u^n};P_{j^n|s^n,x^n,u^n}|P_{u^n,x^n}) > \epsilon\}.
\end{align}
Note that given $(x^n,u^n,j^n)$, the value of
\begin{align}
    I(P_{s^n|x^n,u^n};P_{j^n|s^n,x^n,u^n}|P_{u^n,x^n})
\end{align}
is determined by the joint type $P_{s^n,x^n,u^n,j^n}$. Define a set of types
\begin{align}
    \mathcal{V}(x^n,u^n,j^n)=\{P_{SUXJ}: I(S;J|U,X)>\epsilon, P_{UXJ}=P_{u^n,x^n,j^n},D(P_{SUX}||Q_{SUX})\leq \eta\}.
\end{align}
Then, given the event $\mathcal{E}_0^c$, it follows that for $(x^n,u^n,j^n)$ with joint type $P_{XUJ}$,
\begin{align}
    &Pr\{S^n \in \mathcal{S}^n_{\epsilon}(x^n,u^n,j^n)|x^n,u^n,j^n\}\\
    &=\sum_{s^n\in\mathcal{S}^n_{\epsilon}(x^n,u^n,j^n)} Pr\{S^n=s^n|x^n,u^n\}\\
    &\leq \sum_{P_{SUXJ}\in\mathcal{V}(x^n,u^n,j^n)}\sum_{s^n\in \mathcal{T}^n_{P_{S|UXJ}}[u^n,x^n,j^n]}Pr\{S^n=s^n|x^n,u^n\} \\
    &\overset{(a)}{\leq} (n+1)^{|\mathcal{X}||\mathcal{U}||\mathcal{J}|}2^{nH(P_{S|UXJ}|P_{UXJ})}2^{-n(H(P_{S|UX}|P_{UX})-\delta)} \\
    \label{ine: probability of sn that needs to be considered}&= 2^{-n( I(S;J|X,U) - \delta')}\overset{(b)}{\leq} 2^{-n(\epsilon-\delta')},
\end{align}
where $(a)$ follows by Lemma \ref{lem: type counting lemma} (type counting lemma),  $D(P_{SUX}$ $||Q_{SUX})\leq \eta$ defined in $\mathcal{V}(x^n,u^n,j^n)$ and Lemma \ref{lem: uniform continuity of entropy} (uniform continuity of entropy), $(b)$ follows by the definition of $\mathcal{V}(x^n,u^n,j^n)$. 

In \cite{csiszar1988capacity}, it is proved that we only needs to consider codewords $X^n$ such that $I(X;J)\leq \epsilon$ as the percentage of codewords that violate this condition is exponentially small. Similarly, by Lemma \ref{lem: type lemma 3}, for a fixed jamming sequence $j^n$ and codeword $x^n$ we have
\begin{align}
    \label{neq: number of un s.t. I(UJX)>epsilon}\frac{1}{2^{nR_{\widetilde{S}}}}\Big| \left\{ i:(u^n_i,x^n,j^n) \in \cup_{I(U;J|X)>\epsilon} \mathcal{T}^n_{P_{XJU}}\right\}  \Big| \leq 2^{-n\frac{\epsilon}{3}}.
\end{align}
On account of \eqref{ine: probability of LK}, \eqref{ine: probability of sn that needs to be considered}, \eqref{neq: number of un s.t. I(UJX)>epsilon} and \eqref{ine: overall error probability of stricly causal average error}, it is sufficient to consider $u^n$ and $s^n$ such that $I(U;J|X)<\epsilon$ and $I(S;J|X,U)<\epsilon$. Together with $I(X;J)<\epsilon$ given in \cite{csiszar1988capacity} we have 
\begin{align}
    \label{ine: upper bound of i(usx;j)}I(U,S,X;J) < 3\epsilon
\end{align}
with high probability. This shows 
\begin{align}
    Pr\{\mathcal{E}_1|\mathcal{E}_0^c\} \to 0.
\end{align}

 \subsubsection{Error event $\mathcal{E}_2$}
Given $x^n(m,l),b$ (the bin index of $U^n$ that is decoded in the next block) and the jamming sequence, the error probability is upper bounded by
\begin{align}
    e(x^n(m,l),b,j^n)=\sum_{k} Pr\{K=k\} \sum_{s^n}\sum_{y^n:\phi(y^n)\neq k}W^n(y^n|x^n,s^n,j^n)Pr\{S^n=s^n|x^n,u^n(m,l,b,k)\}
\end{align}
For any joint type $P_{y^n,x^n,u^n,j^n}\notin \Psi^3_{\eta}$, it follows that
\begin{align}
    &\sum_{P_{YXSUJ}\notin{\Psi^3_{\eta}}}\sum_{s^n}\sum_{y^n\in\mathcal{T}^n_{P_{Y|XSUJ}}[x^n,s^n,u^n,j^n]}W_{Y|XSJ}(y^n|x^n,s^n,j^n)Pr\{s^n|x^n,u^n\}\\
    \label{ine: joint distribution violates}&\leq (n+1)^{|\mathcal{X}||\mathcal{U}||\mathcal{J}||\mathcal{Y}|}(\sum_{s^n\notin \mathcal{S}^n_{\epsilon}(x^n,u^n,j^n)}2^{-n (D(P_{u^n,x^n,s^n,j^n,y^n}||W_{Y|XSJ}P_{u^n,x^n,s^n,j^n})-\epsilon)}Pr\{s^n|x^n,u^n\}+2^{-n(\epsilon-\delta')}),
\end{align}
where the inequality is due to Lemma \ref{lem: type fact 3} and \eqref{ine: probability of sn that needs to be considered}.
It then follows that for $s^n$ such that $I(U,X,S;J)< 3\epsilon$,
\begin{align}
    &D(P_{u^n,x^n,s^n,j^n,y^n}||W_{Y|XSJ}P_{u^n,x^n,s^n,j^n})\\
    &=D(P_{u^n,x^n,s^n,j^n,y^n}||W_{Y|XSJ}P_{u^n,x^n,s^n}P_{j^n}) - I(U,X,S;J)\\
    &\overset{(a)}{\geq} D(P_{u^n,x^n,s^n,j^n,y^n}||W_{Y|XSJ}P_{u^n,x^n,s^n}P_{j^n}) - 3\epsilon\\
    &=D(P_{u^n,x^n,s^n,j^n,y^n}||W_{Y|XSJ}Q_{UXS}P_{j^n}) - D(P_{u^n,x^n,s^n}||Q_{UXS}) - 3\epsilon\\
    &\overset{(b)}{\geq} D(P_{u^n,x^n,s^n,j^n,y^n}||W_{Y|XSJ}Q_{UXS}P_{j^n}) - 4\epsilon\\
    &\overset{(c)}{\geq} D(P_{u^n,x^n,j^n,y^n}||Q_{YU|XJ}Q_{X}P_{j^n}) - 4\epsilon \overset{(d)}{\geq} \eta-4\epsilon,
\end{align}
where $(a)$ is due to \eqref{ine: upper bound of i(usx;j)}, $(b)$ is by the definition of $\Psi_{\eta}^2$, $(c)$ follows by projecting the KL divergence to space $\mathcal{U}\times\mathcal{X}\times\mathcal{Y}\times\mathcal{J}, (d)$ is by the definition of $\Psi^3_{\eta}$. It then follows that \eqref{ine: joint distribution violates} is upper bounded by
\begin{align}
    (n+1)^{|\mathcal{X}||\mathcal{U}||\mathcal{J}||\mathcal{Y}|}(2^{-n(\eta-4\epsilon)}+2^{-n(\epsilon - \delta')}).
\end{align}
Thus, we have
\begin{align}
    Pr\{\mathcal{E}_2|\mathcal{E}_0^c,\mathcal{E}_1^c\}\to 0.
\end{align}
 \subsubsection{Error event $\mathcal{E}_3$} 

Now suppose the decoding of $X^n(m,l)$ is correct and the decoder finds a wrong description $u^n(m,l,b,t)$ ($u^n_t$ for short), which implies the joint type $P_{x^n,u^n_t,{j^n}',y^n}\in\Psi^3_{\eta}$ for some ${j^n}'$ and $I(U,S,Y;U'|J,X)>\eta$, where $S$ is defined by some $s^n$ such that $P_{x^n,u^n,s^n,j^n}\in\Psi^4_{\eta}$. Furthermore, there exists some ${S^n}'$ such that $P_{x^n,u^n_t,{s^n}',{j^n}'}\in \Psi^4_{\eta}.$ Consider the joint type $P_{u^n,u^n_t,x^n,s^n,j^n,y^n}$ that satisfies the above conditions. Note that it is sufficient to consider the joint type such that
\begin{align}
    \label{ine: constraint uu'jx}I(U;U',J|X) \leq |R_S'-I(U';J|X)|^++\epsilon.
\end{align}
Otherwise, by Point 3 of Lemma \ref{lem: type lemma 3} we have
\begin{align}
    2^{-nR_S'}\Big| t: (u^n(m,l,b,k),u^n(m,l,b,t),x^n(m,l),j^n) \in \mathcal{T}^n_{P_{UU'XJ}} \;\text{for some $t\neq k$} \Big| < 2^{-n\frac{\epsilon}{2}},
\end{align}
which implies the percentage of such wrong codewords is negligible and the corresponding contribution to the error probability is also negligible.
In the following, for simplicity, we denote $u^n(m,l,b,k)$ by $u^n_k$ and $u^n(m,l,b,t)$ by $u^n_t$.

Then, the probability of estimating a wrong index can be upper bounded by
\begin{align}
    &\sum_{t: (u^n_t,u^n_k,x^n,j^n)\in \mathcal{T}^n_{P_{UU'XJ}}}\sum_{s^n\in \mathcal{T}^n_{P_{S|UU'XJ}}[u^n_k,u^n_t,x^n,j^n]}\sum_{y^n\in\mathcal{T}^n_{P_{Y|UU'XSJ}}[u^n_k,u^n_t,x^n,s^n,j^n]}\\
    &\quad\quad\quad\quad\quad \cdot W^n_{Y|XSJ}(y^n|x^n,s^n,j^n)Pr\{s^n|x^n,u^n_k\}\\
    &\overset{(a)}{\leq} 2^{n(|R_S'-I(U';U,J|X)|^+ + 2\epsilon)} 2^{nH(Y,S|U,U',X,J)} 2^{-nH(Y|U,X,S,J))}2^{-n(H(S|X,U,J)-\delta)}\\
    &=2^{-n(I(Y,S;U'|X,U,J)-|R_S'-I(U';U,J|X)|^+ -\delta - 2\epsilon)}
\end{align}
where (a) follows that $W^n_{Y|XSJ}(y^n|x^n,s^n,j^n)$ is a constant for $y^n\in\mathcal{T}^n_{P_{Y|UU'XSJ}}[u^n_k,u^n_t,x^n,s^n,j^n]$, Point 1 of Lemma \ref{lem: type lemma 3} for the first summation, Lemma \ref{lem: type fact 1} for the second and third summations, the definition of $\Psi^4_{\eta}$, uniform continuity of entropy and condition decreases entropy.

When $R_S' \leq  I(U';J|X)\leq I(U';U,J|X)$, we have $I(U;U'|X,J)\leq I(U;U',J|X)\leq \epsilon$ and hence
\begin{align}
    &I(Y,S;U'|X,U,J)-|R_S'-I(U';U,J|X)|^+ \\
    &= I(Y,S;U'|X,U,J)=I(Y,S,U;U'|X,J) - I(U;U'|X,J) \geq \eta - \epsilon.
\end{align}

When $R_S'>I(U';J|X),$ it follows from \eqref{ine: constraint uu'jx} that
\begin{align}
    R_S' &\geq I(U;U',J|X) + I(U';J|X) - \epsilon\\
    &=I(U,J;U'|X) + I(U;J|X)-\epsilon \geq I(U,J;U'|X) - \epsilon, 
\end{align}
and hence
\begin{align}
    |R_S'-I(U';U,J|X)|^+ \leq R_S'-I(U';U,J|X)+\epsilon.
\end{align}
Now, we can bound the error probability by
\begin{align}
    &2^{-n(I(Y,S;U'|X,U,J)-|R_S'-I(U';U,J|X)|^+ -\delta - 2\epsilon)}\\
    &\leq 2^{-n(I(Y,S;U'|X,U,J)-R_S'+I(U';U,J|X) -\delta - 3\epsilon)}\\
    &=2^{-n(I(Y,S,U,J;U'|X)-R_S' - \delta - 3\epsilon)}\\
    &\leq 2^{-n(I(Y;U'|X)-R_S' - \delta - 3\epsilon)},
\end{align}
which vanishes if $R_S' \leq I(Y;U'|X).$ By the definition of the joint type $P_{u^n,u^n_t,x^n,s^n,j^n,y^n}$, there exist $S',J'$ defined by ${s^n}',{j^n}'$ such that 
$P_{x^n,u^n_t,{s^n}',{j^n}'}\in \Psi^4_{\eta}$ and $P_{x^n,u^n_t,y^n,{j^n}'}\in \Psi^3_{\eta}$. By the continuity of mutual information, we have
\begin{align}
    |I(U';Y|X)-I_Q(U;Y|X)|\leq \eta,
\end{align}
where the subscript $Q$ indicates the mutual information is computed according to the joint distribution $Q_{USXY|J}$ we defined for Theorem \ref{the: achievability of strictly causal}. Consider the worst case, the error probability vanishes if
\begin{align}
    R_S' \leq \min_{Q_{J}\in\mathcal{P}(\mathcal{J})}\min_{P_S\in\mathcal{P}(\mathcal{S},Q_S,\delta_n)} I_Q(U;Y|X),
\end{align}
where 
\begin{align}
    \mathcal{P}(\mathcal{S},Q_S,\delta_n)=\{P_S\in\mathcal{P}(\mathcal{S}):D(P_S||Q_S)\leq \delta_n\}
\end{align}
and $\delta_n\to 0$ as $n\to \infty.$
When the block length $n$ is sufficiently large, by the law of large numbers, we have
\begin{align}
    R_S' \leq \min_{Q_{J}\in\mathcal{P}(\mathcal{J})} I_Q(U;Y|X).
\end{align}

Once the receiver decodes $u^n$, it sets $\hat{s}_i = h(x_i,u_i,y_i)$ and the distortion constraint is satisfied by the typical average lemma\cite{el2011network}.

\section{proof of corollary \ref{coro: strictly causal lossless average error} and corollary \ref{coro: noncausal pure lossless state com}}

\subsection{Proof of Corollary \ref{coro: strictly causal lossless average error}}\label{sec: proof of coro: strictly causal lossless average error}
The achievability follows by setting $U=S$ in Theorem \ref{the: achievability of strictly causal}. When the state observation has a delay $d=n$, the encoder just needs to wait for one additional initial block to collect the state information. When the number of blocks is sufficiently large, the additional cost is negligible. For the converse, we first assume that the channel is a compound channel. By the definition of lossless state communication, we have
    \begin{align}
        H_{Q_J}(S^n|Y^n,X^n)\leq H_{Q_J}(S^n|Y^n)\leq \epsilon,\;\;\forall Q_J\in\mathcal{P}(\mathcal{J}).
    \end{align}
    It follows that
    \begin{align}
        nR &= H(M)\\
        &\overset{(a)}{\leq} H(M) - H_{Q_J'}(S^n|Y^n,X^n) + \epsilon\\
            &\overset{(b)}{\leq} I_{Q_J}(M;Y^n) - H_{Q_J'}(S^n|Y^n,X^n) + 2\epsilon\\
            &=\sum_{i=1}^n I_{Q_J}(M;Y_i|Y^{i-1})-  H_{Q_J'}(S^n|Y^n,X^n) + 2\epsilon\\
            &\leq \sum_{i=1}^n I_{Q_J}(M,X_i;Y_i|Y^{i-1})-  H_{Q_J'}(S^n|Y^n,X^n) + 2\epsilon\\
            &\overset{(c)}{\leq}\sum_{i=1}^n I_{Q_J}(X_i;Y_i)-  H_{Q_J'}(S^n|Y^n,X^n) + 2\epsilon,
    \end{align}
    where $(a)$ and $(b)$ follow from Fano's inequality, $(c)$ follows from the fact that conditions reduce entropy and the Markov chain $Y_i-X_i-(M,Y^{i-1})$.
    
    For the lossless communication rate, we have
    \begin{align*}
        &H_{Q_J'}(S^n|Y^n,X^n)\\
        &\overset{(a)}{=}H(S^n) - I_{Q_J'}(S^n;Y^n|X^n)\\
        &=H(S^n) - \sum_{i=1}^n I(S^n;Y_i|Y^{i-1},X^n)\\
        &=\sum_{i=1}^n H(S_i) - H(Y_i|Y^{i-1},X^n) + H(Y_i|X^n,S^n,Y^{i-1})\\
        &\overset{(b)}{\geq} \sum_{i=1}^n H(S_i) - H(Y_i|X_i) + H(Y_i|X_i,S_i)\\
        &=\sum_{i=1}^n H(S_i) - I(S_i;Y_i|X_i) \overset{(c)}{=} \sum_{i=1}^n H(S_i|Y_i,X_i).
    \end{align*}
    where $(a)$ and $(c)$ are due to our observation delay assumption, so that $S^n$ is independent of $X^n$, $(b)$ is by the Markov chain $Y_i-(X_i,S_i)-(X^{n\backslash i},S^{n\backslash i}, Y^{i-1})$.
     Now, we have
    \begin{align}
        nR &\leq \sum_{i=1}^n I_{Q_J}(X_i;Y_i)-  H_{Q_J'}(S_i|Y_i,X_i) + 2\epsilon\\
        &\leq \max_{Q_{X}}I_{Q_J}(X;Y)-  H_{Q_J'}(S|Y,X) + 2\epsilon
    \end{align}
    This completes the proof.
    
\subsection{Proof of Corollary \ref{coro: noncausal pure lossless state com}}\label{sec: proof of coro: noncausal pure lossless state com}
\emph{Achievability:} Fix an input distribution $Q_{X|S}$. The encoder partitions the state observation into two parts $s^{n_1}s^{n_2}$ such that $n_1+n_2=n, n_1/n_2\to 0$ as $n\to\infty$. For each type $P_S$ of $S^{n_2}$ and $S^{n_2}\in T^{n_2}_{P_S}$, construct a codebook such that each $S^{n_2}\in T^{n_2}_{P_S}$ is associated with a codeword $X^{n_2}$. Specifically, for each $s^{n_2}\in T^{n_2}_{P_S}$, select a codeword $x^{n_2}$ from $T^n_{Q_{X|S}}[s^{n_2}]$ uniformly at random. Further, construct a Gel'fand--Pinsker type code $\{u^{n_1}\}$ for the set of all possible types of $S^{n_2}$. Since the number of possible types for $n_2-$length sequences is a polynomial number with respect to $n_2$, the rate of this prefixed code is negligible. It can be reliably transmitted to the receiver if there exist distributions $Q_{U|S},Q_{X|US}$ such that the channel $Q_{Y|UJ}$ is nonsymmetrizable-$\mathcal{U}.$

    After decoding the type of the sequence, the decoder tries to find a unique pair $(S^{n_2},X^{n_2})$. The decoding rule is the same as the arbitrarily varying multiple access channels, and only the sum rate constraint needs to be considered since each $s^{n_2}$ is associated with a fixed codeword $x^{n_2}$. The nonsymmetrizable-$\mathcal{X}\times\mathcal{S}$ condition ensures the reliability of the decoding. On the other hand, if the channel $W_{Y|XSJ}$ is nonsymmetrizable-$\mathcal{X}\times\mathcal{S}$, there always exists at least one $Q_{XU|S}$ such that the channel $Q_{Y|UJ}$ is nonsymmetrizable-$\mathcal{U}$ since the encoder can always set $U=(X,S)$. Hence, the nonsymmetrizability regarding $\mathcal{X}\times\mathcal{S}$ is sufficient. This completes the proof of the achievability.

    \emph{Converse:} To show the converse, we first assume the channel is a compound channel, i.e. the distribution $Q_J$ of the jammer is fixed during each transmission block. By the lossless estimation constraint, we have by Fano's inequality,
    \begin{align}
        H(S^n|Y^n) \leq n\epsilon\;\;\text{for any $Q_J \in \mathcal{P}(\mathcal{J})$}
    \end{align}
    for some $\epsilon > 0.$ It follows that for any $Q_J \in \mathcal{P}(\mathcal{J})$,
    \begin{align*}
        &nH(S) \\
        &= H(S^n) \\
        &= H(S^n|Y^n) + I(S^n;Y^n)\\
        &\leq I(S^n;Y^n) + n\epsilon \\
        &\leq I(X^n,S^n;Y^n)+ n\epsilon\\
        &\leq \sum_{i=1}^n I(X_i,S_i;Y_i)+ n\epsilon\\
        &\leq n \max_{Q_{X|S}} I(X,S;Y)+ n\epsilon.
    \end{align*}
    The proof is completed by the fact that $I(X,S;Y)$ is convex in $Q_J$ and concave in $Q_{X|S}$.

\section{Proof of Theorem \ref{the: noncausal maximal error}}\label{sec: proof of achievability of noncausal case}
Fix a joint distribution $Q_{SU}Q_{X|US}W_{Y|XSJ}$ and define
\begin{align}
    \Psi^1_{\eta} = \{P_{u^n,y^n,j^n}: D(P_{y^n|u^n,j^n}||Q_{Y|UJ})\}\leq \eta,\\
    \Psi^2_{\eta} = \{P_{u^n,s^n,j^n}:D(P_{u^n,s^n,j^n}||Q_{US}P_{J|US})\}.
\end{align}
Further set positive real numbers $\widetilde{R},R$ and $R'$ such that
\begin{align}
    &\widetilde{R} = \min\{\min_{P_J\in\mathcal{P}(\mathcal{J|U})}I(U;Y),D(Q_U)\} - \tau,\\
    &R':=\widetilde{R}-R \geq I(U;S) + \tau.
\end{align}

\subsection{Codebook Generation}
 This section introduces the generation of codebooks and some important properties of randomly generated codebooks. The first property, proved by Lemma \ref{lem: good codebook}, is an auxiliary result and upper bounds the number of codewords in the codebook that fall in the conditional type set of an arbitrary sequence and a condition $V-$shell. The second property, proved by Lemma \ref{lem: size of subcodebooks}, shows that after the random binning operation, most subcodebooks have sufficiently large sizes to cover the typical space of the channel state sequence $S^n$. The last property in Lemma \ref{lem: good codeword} ensures that most codewords in the codebook are `good' codewords whose definition will be given later. Intuitively, `good' codewords are distinguishable codewords that will not cause any ambiguity at the decoder side, using the decoding rule we define later.

Select $2^{n\widetilde{R}}$ codewords $\{u^n\}$ uniformly at random from a fixed type set $\mathcal{T}^n_{Q_U}$, where $Q_U = \sum_{s}Q_SQ_{U|S}$. Assign each codeword an index $m\in[1:2^{nR}]$ such that codewords with the same index form a bin $\mathcal{C}(m)=\{u^n|m\}:=\{u^n(l):l\in [1:2^{n(\widetilde{R}-R)}]\}$. Denote the index assigned to a codeword $u^n$ by $b(u^n)$.
With this random assignment of $m\in [1:2^{nR}]$, we have a set of sub-codebooks $\{u^n|m\},m\in[1:2^{nR}]$, each with size $2^{nR'},R':=\widetilde{R}-R.$  In the following, we prove a property of the codebook.
\begin{applemma}\label{lem: good codebook}
    For given distribution $Q_{U}$ and $R'\geq 0$,  we call any of the above codebook $\{u^n\}$ a `good' codebook if for every $j^n\in\mathcal{J}^n$ and $V:\mathcal{J}\to\mathcal{U}$ we have
    \begin{align}
        |k:u^n(k)\in\mathcal{T}^n_{V}[j^n]| \leq 2^{n(|\widetilde{R}-I(P_{j^n};V)|^+ + \delta)}.
    \end{align}
    Then, it follows that
    \begin{align}
        Pr\{\{U^n\} \;\text{is a good codebook}\} \geq 1-\zeta,
    \end{align}
    where $\zeta$ is a double exponentially small number with respect to $n$.
\end{applemma}
\begin{proof}
    Fix an arbitrary $j^n$ and a conditional shell $V$. Define an indicator function
    \begin{align}
        \mathbb{I}_{j^n}\{k\} = \left\{
            \begin{aligned}
                &1, \;U^n(k)\in \mathcal{T}^n_{V}[j^n],\\
                &0, \;\text{otherwise.}
            \end{aligned}
        \right. 
    \end{align}
    Note that for each sub-codebook, the codewords $u^n$ are all randomly selected from the given conditional type set $\mathcal{T}^n_{Q_{U}}$. Hence, we have
    \begin{align*}
        &\mathbb{E}[\mathbb{I}_{j^n}\{k\}] = Pr\{U^n(k)\in \mathcal{T}^n_{V}[j^n]\}\\
        &=\sum_{u^n\in \mathcal{T}^n_{V}[j^n]} \frac{1}{|\mathcal{T}^n_{Q_{U}}|} = 2^{-nI(P_{j^n};V)}.
    \end{align*}
    Then we have
    \begin{align}
        &Pr\{|k:U^n(k)\in\mathcal{T}^n_{V}[j^n]| > 2^{n(|\widetilde{R}-I(P_{j^n};V)|^+ + \delta)}\}\notag\\
        &=Pr\{\sum_{k}\mathbb{I}_{j^n}\{k\}  > 2^{n(|\widetilde{R}-I(P_{j^n};V)|^+ + \delta)}\}\notag\\
        &\overset{(a)}{\leq} \exp_2(-2^{n(|\widetilde{R}-I(P_{j^n};V)|^+ + \delta)})\prod_{k}\mathbb{E}[\exp_2(\mathbb{I}_{j^n}\{k\})]\notag\\
        &\overset{(b)}{\leq} \exp_2(-2^{n(|\widetilde{R}-I(P_{j^n};V)|^+ + \delta)})\exp_2(2^{n(\widetilde{R}-I(P_{j^n};V))}\log e)\notag\\
        \label{eq: good codebook bound}&=\exp_2(-2^{n(|\widetilde{R}-I(P_{j^n};V)|^+ + \delta)}+2^{n(\widetilde{R}-I(P_{j^n};V))}\log e),
    \end{align}
    where $(a)$ is by the Markov inequality, $(b)$ follows by the inequality $\mathbb{E}[2^T]\leq \mathbb{E}[1+T]=1+\mathbb{E}[T]\leq \exp_e(\mathbb{E}[T])$ for $0\leq T \leq 1.$
    
    When $\widetilde{R}\geq I(P_{j^n};V)$, \eqref{eq: good codebook bound} is upper bounded by
    \begin{align}
        \exp_2(-2^{n(\widetilde{R}-I(P_{j^n};V))}(2^{n\delta}-\log e))
    \end{align}
    and when $\widetilde{R}< I(P_{j^n};V)$,it is upper bounded by
    \begin{align}
        \exp_2(-2^{n\delta}+\log e).
    \end{align}
    It follows that 
    \begin{align}
        &Pr\{|k:U^n(k)\in\mathcal{T}^n_{V}[j^n]| > 2^{n(|\widetilde{R}-I(P_{j^n};V)|^+ + \delta)}\} \leq \exp_2(-2^{n\delta}+\log e):=\zeta
    \end{align}
    and the proof is completed by noting that the number of all possible types is upper-bounded by $(n+1)^{|\mathcal{U}||\mathcal{J}|}$ and the number of all $j^n$ is only exponentially large with respect to $n$.
\end{proof}
As we mentioned before, we use the random binning technique to the codebook instead of deterministic binning like we usually do for the Gel'fand-Pinsker coding. Here, we prove that most of the subcodebooks are still large enough to cover the typical space of the channel state sequences.
\begin{applemma}\label{lem: size of subcodebooks}
    For a random message $M$, the size of the corresponding sub-codebook satisfies
    \begin{align}
        \label{ine: probability of large enough codebook}Pr\{|\{U^n|M\}| \geq 2^{n(I(U;S)+ \tau')} \} \geq 1 - \zeta, 
    \end{align}
    for some $\tau'>0$, where $\zeta$ is a double-exponentially small number.
\end{applemma}
\begin{proof}
    Define an indicator function
    \begin{align*}
        \mathbb{I}_M(U^n)=\left\{
        \begin{aligned}
            &1\;\;\;b(U^n) = M;\\
            &0\;\;\;\text{otherwise.}
        \end{aligned}
        \right.
    \end{align*}
    It follows that
    \begin{align}
        |\{U^n|M\}| = \sum_{i=1}^{2^{n\widetilde{R}}} \mathbb{I}_M(U^n(i))
    \end{align}
    and
    \begin{align}
        \mathbb{E}[ \mathbb{I}_M(U^n(i))] = 2^{-nR}.
    \end{align}
    Therefore, it follows that
    \begin{align}
        \mathbb{E}[|\{U^n|M\}|] = \mathbb{E}[\sum_{i=1}^{2^{n\widetilde{R}}} \mathbb{I}_M(U^n(i))] = 2^{nR'}.
    \end{align}
    Applying the Chernoff bound for the sum of binary random variables gives
    \begin{align}
        Pr\{|\sum_{i=1}^{2^{n\widetilde{R}}} \mathbb{I}_M(U^n(i)) - 2^{nR'}|>2^{-n\epsilon}2^{nR'}\} \leq 2e^{-2^{nR'}2^{-2n\epsilon}/3}:=\zeta,
    \end{align}
    where $\epsilon>0$ is a positive real number such that $\tau > 2\epsilon$. Hence, with a probability greater than $1-\zeta$, we have
    \begin{align}
        |\{U^n|M\}| \geq (1-2^{-n\epsilon})2^{nR'} = 2^{n(I(U;S)+\tau')}
    \end{align}
    by choosing a sufficiently small $\epsilon$.
\end{proof}
In the codebook generation, we partition the codebook into $2^{nR}$ subcodebooks $\{u^n|m\},m\in[1:2^{nR}]$. By Lemma \ref{lem: size of subcodebooks} and Bernoulli's inequality, the probability that all subcodebooks satisfying \eqref{ine: probability of large enough codebook} is lower bounded by
\begin{align}
    (1-\zeta)^{2^{nR}} \geq 1 - \zeta \cdot 2^{nR},
\end{align}
where $\zeta$ is a double exponentially small number. Hence, we can neglect the events where there exist subcodebooks whose sizes are not large enough. 

Before we delve into the error analysis, it remains to show that most codewords in the codebook are `good' codewords, which are defined as follows.

\textbf{\emph{Good codeword verification:}}
For a given codebook $\mathcal{C}$ with size $2^{n\widetilde{R}}$, we define a codeword $u^n(i)$ to be a good codeword if
\begin{align}
    \label{def: good codeword}I(u^n(i);u^n(j)) <\widetilde{R} 
\end{align}
for any $i\neq j, u^n_i,u^n_j\in\mathcal{C}$. In the following, we show that most of the codewords of a randomly generated codebook are good codewords with high probability.

\begin{applemma}\label{lem: good codeword}
    Given finite sets $\mathcal{U},\mathcal{J}$ and the distribution $Q_{U}$, for any random codebook $\textbf{C}$ of size $2^{n\widetilde{R}},\widetilde{R}=\min\{I(Q_{U}),D(Q_{U})\}$, generated by performing random selection over the type set $\mathcal{T}^n_{Q_{U}}$, we have
    \begin{align}
        Pr\{\text{The number of good codewords $\geq (1-2^{-n\delta})|\textbf{C}|$}\} > 1 - \chi,
    \end{align}
    where $\chi>0$ is a double exponentially small number with respect to $n$.
\end{applemma}
By Lemma \ref{lem: good codebook}, with high probability, the codebook has a set of codewords $U^n$ such that
\begin{align}
    \label{eq: refinement 1}&|k:U^n(k)\in \mathcal{T}^n_{V}[u^n,j^n]| \leq 2^{n(|\widetilde{R} - I(P_{u^n,j^n};V)|^++\delta)}\;\text{for every $V:\mathcal{U}\times\mathcal{J}\to\mathcal{U}$},\\
    \label{eq: refinement 2}&|k:U^n(k)\in \mathcal{T}^n_{\Bar{V}}[u^n]| \leq 2^{n(|\widetilde{R} - I(P_{u^n};\Bar{V})|^++\delta)}\;\text{for every $\Bar{V}:\mathcal{U}\to\mathcal{U}$},
\end{align}
where $\delta\to 0$ as $n\to \infty$.
Now by \eqref{eq: refinement 2}, it follows that for every $i\in[1:2^{n\widetilde{R}}]$,
\begin{align}
    |k:U^n(k)\in \mathcal{T}^n_{\Bar{V}}[U^n(i)]| \leq 2^{n\delta},\;\;\text{if $\widetilde{R} \leq I(P_{U^n};\Bar{V})$}
\end{align}
and hence by the type counting lemma,
\begin{align}
    |k:I(U^n(k);U^n(i)) \geq \widetilde{R}|\leq 2^{n\delta}(n+1)^{|\mathcal{U}|^2}=2^{n\hat{\delta}}.
\end{align}
where $\hat{\delta}\to 0$ as $n\to \infty$.
Define 
\begin{align}
    \mathcal{B}(U^n(i))=\{U^n(k)\in\textbf{C}:I(U^n(k);U^n(i)) \geq \widetilde{R}\}.
\end{align}
It follows that
\begin{align}
    |\sum_{i=1}^{|\textbf{C}|}\mathcal{B}(U^n(i))|\leq 2^{n(\widetilde{R}+\hat{\delta})}
\end{align}
and
\begin{align}
    \label{ine: probability upper bound of bui}Pr\left\{\sum_{i=1}^{|\textbf{C}|}\mathcal{B}(U^n(i))\right\}\leq 2^{-nH(U)}2^{n(\widetilde{R}+\hat{\delta})}\overset{(a)}{\leq}2^{-n(\max\{H(U|Y),H(U|U')\}-\hat{\delta})},
\end{align}
where $(a)$ follows by the fact that $\widetilde{R}\leq \min\{I(Q_{U}),D(Q_{U})\}$. Now for a randomly generated codebook $\textbf{C}$, we define an indicator function
\begin{align}
    \mathbb{I}_i = \left\{ 
        \begin{aligned}
            &1 \;\;\text{if $U^n(i)$ is not a good codeword.}\\
            &0 \;\;\text{if $U^n(i)$ is a good codeword.}
        \end{aligned}
    \right.
\end{align}
It follows that
\begin{align}
    Pr\{\text{The number of good codewords} < (1-2^{-n\nu})|\mathcal{C}|\} = Pr\{\sum_{i}\mathbb{I}_i \geq 2^{-n\nu}|\mathcal{C}|\}
\end{align}
Note that the sum $\sum_{i}\mathbb{I}_i$ can be upper bounded by
\begin{align}
    \sum_{i}\mathbb{I}^1_i + \sum_{i}\mathbb{I}^2_i
\end{align}
where
\begin{align}
    \mathbb{I}^1_i = \left\{
        \begin{aligned}
            &1, \;\;\text{if $U^n(i)\in \cup_{j<i}\mathcal{B}(U^n(j))$},\\
            &0  \;\;\text{otherwise}
        \end{aligned}
    \right.
\end{align}
and
\begin{align}
    \mathbb{I}^2_i = \left\{
        \begin{aligned}
            &1, \;\;\text{if $U^n(i)\in \cup_{j>i}\mathcal{B}(U^n(j))$},\\
            &0  \;\;\text{otherwise}
        \end{aligned}
    \right.
\end{align}
Now, we have
\begin{align}
    Pr\{\sum_{i}\mathbb{I}_i \geq 2^{-n\nu}|\mathcal{C}|\} \leq Pr\{\sum_{i}\mathbb{I}^1_i \geq \frac{1}{2}2^{-n\nu}|\mathcal{C}|\} + Pr\{\sum_{i}\mathbb{I}^2_i \geq \frac{1}{2}2^{-n\nu}|\mathcal{C}|\}
\end{align}
and it suffices to bound only $Pr\{\sum_{i}\mathbb{I}^1_i \geq \frac{1}{2}2^{-n\nu}\}.$ Here, we use the following lemma:
\begin{applemma}{\cite[Lemma 1(a)]{ahlswede1980method}}
    Let $T_1,T_2,...,T_K$ be a sequence of discrete random variables, then
    \begin{align}
        Pr\{\frac{1}{K}\sum_{i=1}^K T_i \geq a\} \leq 2^{-\alpha a K}\prod_{i=1}^K \max_{t^{i-1}}\mathbb{E}[\exp (\alpha T_i)|t^{i-1}]
    \end{align}
\end{applemma}
By defining $K=2^{n\widetilde{R}}, T_i = \mathbb{I}_i^1, a = \frac{1}{2}2^{-n\nu}, \alpha = 1$ and \eqref{ine: probability upper bound of bui}, we have
\begin{align}
    Pr\{T_i=1|T_1=t_1,T_2=t_2,...,T_{i-1}=t_{i-1}\} \leq 2^{-n(\max\{H(U|Y),H(U'|U)\}-\hat{\delta})}
\end{align}
and
\begin{align}
    Pr\{\sum_{i}\mathbb{I}^1_i \geq \frac{1}{2}2^{-n\nu}|\textbf{C}|\}&\overset{(a)}{\leq} 2^{-\frac{1}{2}2^{-n\nu}|\textbf{C}|}[1+2^{-n(\max\{H(U|Y),H(U'|U)\}-\hat{\delta})}]^{|\textbf{C}|}\\
    &\overset{(b)}{\leq} 2^{-\frac{1}{2}2^{-n\nu}|\textbf{C}|} e^{2^{-n(\max\{H(U|Y),H(U'|U)\}-\hat{\delta})}|\textbf{C}|}\\
    &=2^{-\frac{1}{2}2^{-n\nu}|\textbf{C}| + 2^{-n(\max\{H(U|Y),H(U'|U)\}-\hat{\delta})}|\textbf{C}|\log e}\\
    &=2^{-|\textbf{C}|(\frac{1}{2}2^{-n\nu}-2^{-n(\max\{H(U|Y),H(U'|U)\}-\hat{\delta})}\log e)}
\end{align}
where $(a)$ and $(b)$ follow by the inequality $2^t\leq 1+t \leq e^t$ for $t\in[0,1]$.

Now we argue the value of $\max\{H(U|Y),H(U|U')\}$. When $\max\{H(U|Y),H(U|U')\}=0$, we have both $I(Q_{U})=D(Q_{U})=H(Q_{U})$. By the definition of $\widetilde{R}$ we have 
\begin{align}
    I(u^n(i),u^n(j)) < \widetilde{R}
\end{align}
always hold for any $i\neq j$, and all the codewords are good codewords. When $\max\{H(U|Y),H(U|U')\}>0$, there must exists a $\zeta>0$ independent of $n$ such that $\max\{H(U|Y),H(U|U')\} > \zeta$. Then we choose a sufficiently large $n$ such that $\nu+\hat{\delta}<\zeta$ and it follows that
\begin{align}
    Pr\{\sum_{i}\mathbb{I}^1_i \geq \frac{1}{2} 2^{-n\nu}|\textbf{C}|\}&\leq 2^{-|\textbf{C}|(\frac{1}{2}2^{-n\nu}-2^{-n(\max\{H(U|Y),H(U'|U)\}-\hat{\delta})}\log e)}\\
    &\leq 2^{-|\textbf{C}|(\frac{1}{2} 2^{-n\nu}-2^{-n(\zeta-\hat{\delta})}\log e)}\\
    &\leq 2^{-|\textbf{C}|2^{-n\hat{\nu}}}:=\chi
\end{align}
for some $\hat{\nu}>0$.
This completes the proof. \hfill \qedsymbol

Let $\mathcal{C}_G\subseteq \mathcal{C}$ be the set of good codewords of the codebook $\mathcal{C}$. 
\begin{applemma}\label{lem: number of bad codewords in each subcodebook}
    For a random message $M$, the number of bad codewords in the subcodebook $\mathcal{C}(M)$ satisfies
    \begin{align}
        Pr\{|(\mathcal{C}\backslash \mathcal{C}_G) \cap \mathcal{C}(M)| \leq 2^{-2n\epsilon}2^{nR'}\} \geq 1 - \chi,
    \end{align}
    for some positive $\epsilon$, where $\chi$ is a double-exponentially small number.
\end{applemma}
Due to the random binning operation to the codebook, Lemma \ref{lem: number of bad codewords in each subcodebook} follows the same as Lemma \ref{lem: size of subcodebooks}.  The lemma ensures that the proportion of bad codewords in each subcodebook is exponentially small.
\subsection{Encoding and Decoding}
\emph{Encoding: } The encoder observes the random state $s^n$ and tries to transmit a message $m$. It looks for a codeword $u^n\in\{u^n|m\}$ such that
\begin{align}
    (u^n,s^n)\in\mathcal{T}^n_{Q_{US},\delta}.
\end{align}
If there are multiple such codewords, the encoder selects one uniformly at random. If there is no such codeword, the encoder declares an error. If the selected codeword is a bad codeword, the encoder declares an error.

\emph{Decoder: }The decoding rule is defined as follows.
Upon observing the channel output $y^n$, the decoder ouputs $\phi(y^n)=(\hat{m},\hat{l})$ if and only if there exists a pair of $(s^n,j^n)\in\mathcal{S}^n \times \mathcal{J}^n$ such that
\begin{itemize}
    \item The joint type $P_{u^n(\hat{m},\hat{l}),y^n,j^n}\in \Psi^1_{\eta}$.
    \item The joint type $P_{u^n(\hat{m},\hat{l}),s^n,j^n}\in\Psi^2_{\eta}$.
    \item For each pair of competitor $(\hat{m}',\hat{l}')\neq (\hat{m},\hat{l})$ such that $P_{u^n(\hat{m}',\hat{l}'),{s^n}',y^n,{j^n}'}\in\Psi^1_{\eta}$ for some ${s^n}'\in\mathcal{S}^n,{j^n}'\in\mathcal{J}^n$, we have $I(Y,S;U'|U,J)\leq \eta$, where $U,U',S,J,Y$ are dummy random variables such that the joint distribution equals $P_{UU'SYJ}$, which is the type of $(u^n(\hat{m},\hat{l}),u^n(\hat{m},\hat{l}),s^n,y^n,j^n)$.
\end{itemize}
The decoder reconstructs the sequence $\hat{s}^n$ by $\hat{s}_i=h(u_i,y_i),i=1,...,n.$
\subsection{Unambiguity Analysis}
Again, we first show that the decoding rule we defined above is unambiguous.
Define a dummy random variable $J$ such that $P_J$ is defined by the type of $j^n$. By the law of large numbers we have with high probability
\begin{align}
    D(P_{s^n,j^n}||Q_SP_J)\leq \eta.
\end{align}
and due to the way that $u^n$ is selected,
\begin{align}
    \label{def: joint type of usj}Pr\{D(P_{u^n,s^n,j^n}||Q_{SU}P_{J|US})\leq \eta\}\to 1.
\end{align}
Since $x^n$ are generated by $(u^n,s^n)$ and $Q_{X|US}$, by the argument for conditional typicality\cite[Appendix 2A]{el2011network} we have
\begin{align}
    \label{def: joint type of xusj}Pr\{D(P_{x^n,u^n,s^n,j^n}||Q_{US}Q_{X|US}P_{J|US})\leq \eta\} \to 1.
\end{align}
Since the channel $W_{Y|XSJ}$ is discrete memoryless, we also have
\begin{align}
    Pr\{D(P_{y^n|x^n,s^n,j^n}||W_{Y|XSJ})\leq\eta\} \to 1
\end{align}
In the following, we show the unambiguity of the decoding rule when a good codeword is selected.
Now we have
\begin{align}
    2\eta&\geq D(P_{y^n|x^n,s^n,j^n}||W_{Y|XSJ}) + D(P_{x^n,u^n,s^n,j^n}||Q_{SU}Q_{X|SU}P_{J|US})\\
    &=\sum_{y,x,s,j} P_{x^n,s^n,j^n,y^n}(x,s,j,y)\log\frac{P_{y^n|x^n,s^n,j^n}(y|x,s,j)}{W_{Y|XSJ}(y|x,s,j)} \\
    &\quad\quad\quad + \sum_{x,u,s,j}P_{x^n,u^n,s^n,j^n}(x,u,s,j)\log\frac{P_{x^n,j^n,u^n,s^n}(x,j,u,s)}{Q_{US}(u,s)Q_{X|US}(x|u,s)P_{j^n|u^n,s^n}(j|u,s)}\\
    &=\sum_{x,s,j,u,y}P_{x^n,s^n,u^n,j^n,y^n}(x,s,j,u,y)\log\frac{P_{y^n|x^n,s^n,j^n}(y|x,s,j)P_{x^n|u^n,s^n,j^n}(x|u,s,j)P_{u^n,s^n,j^n}(u,s,j)}{W_{Y|XSJ}(y|x,s,j)Q_{X|US}(x|u,s)Q_{US}(u,s)P_{j^n|u^n,s^n}(j|u,s)}.
\end{align}
Projecting the above inequality into $\mathcal{S}\times\mathcal{J}\times\mathcal{U}\times\mathcal{Y}$ gives
\begin{align}
    \label{ine: type bound of psjuy}2\eta \geq \sum_{s,j,u,y}P_{s^n,u^n,j^n,y^n}(s,j,u,y)\log\frac{P_{y^n|u^n,s^n,j^n}(y|u,s,j)P_{u^n,s^n,j^n}(u,s,j)}{Q_{Y|USJ}(y|u,s,j)Q_{US}(u,s)P_{j^n|u^n,s^n}(j|u,s)},
\end{align}
where $Q_{Y|USJ}(y|u,s,j) = \sum_{x}W_{Y|XSJ}(y|x,s,j)Q_{X|US}(x|u,s)$. 
Projecting the above expression into $\mathcal{J}\times\mathcal{U}\times\mathcal{Y}$ yields
\begin{align}
    2\eta &\geq \sum_{j,u,y}P_{u^n,j^n,y^n}(j,u,y)\log\frac{P_{y^n,u^n,j^n}(y,u,j)}{\sum_{s}Q_{Y|USJ}(y|u,s,j)Q_{US}(u,s)P_{j^n|u^n,s^n}(j|u,s)}\\
    &=\sum_{j,u,y}P_{u^n,j^n,y^n}(j,u,y)\\
    &\quad\quad\quad\quad \log\frac{P_{y^n,u^n,j^n}(y,u,j)/P_{u^n,j^n}(u,j)}{\sum_{s}Q_{Y|USJ}(y|u,s,j)Q_{US}(u,s)P_{j^n|u^n,s^n}(j|u,s)/\sum_{y,s}Q_{YUSJ}(y,u,s,j)}\frac{P_{u^n,j^n}(u,j)}{\sum_{y,s}Q_{YUSJ}(y,u,s,j)}
\end{align}
where $Q_{YUSJ}(y,u,s,j)=Q_{Y|USJ}(y|u,s,j)Q_{US}(u,s)P_{j^n|u^n,s^n}(j|u,s).$ By \eqref{def: joint type of usj} we have
\begin{align}
    \eta &\geq \sum_{j,u}P_{u^n,j^n}(j,u)\log \frac{P_{u^n,j^n}(u,j)}{\sum_{y,s}Q_{YUSJ}(y,u,s,j)}
\end{align}
and hence 
\begin{align}
    2\eta &\geq \sum_{j,u,y}P_{u^n,j^n,y^n}(j,u,y)\log\frac{P_{y^n,u^n,j^n}(y,u,j)/P_{u^n,j^n}(u,j)}{\sum_{s}Q_{Y|USJ}(y|u,s,j)Q_{US}(u,s)P_{j^n|u^n,s^n}(j|u,s)/\sum_{y,s}Q_{YUSJ}(y,u,s,j)}\\
    \label{neq: pujl qujl}&=\sum_{j,u,y}P_{u^n,j^n,y^n}(j,u,y) \log\frac{P_{y^n|u^n,j^n}(y|u,j)}{Q_{Y|UJ}(y|u,j)},
\end{align}
where $Q_{Y|UJ}(y|u,j)=Q_{YUJ}(y,u,j)/\sum_{y,s}Q_{YUSJ}(y,u,s,j).$ The remaining proof follows similarly to that of the strictly causal case. By the definition of $\Psi^1_{\eta}$
\begin{align}
    D(P_{y^n|u^n,j^n}||Q_{Y|UJ})=\sum_{j,u,y}P_{u^n,j^n,y^n}(j,u,y)\log \frac{P_{y^n|u^n,j^n}(y|u,j)}{Q_{Y|UJ}(y|u,j)}\leq 2\eta 
\end{align}
and
\begin{align}
    I(Y;U'|U,J) = \sum_{u,u',j,y}P_{u^n,{u^n}',j^n,y^n}(u,u',j,y)\log \frac{P_{y^n|u^n,{u^n}',j^n}(y|u,u',j)}{P_{y^n|u^n,j^n}(y|u,j)}\leq I(Y,S;U'|J,U)\leq \eta.
\end{align}
It follows that
\begin{align}
    3\eta \geq &\sum_{u,u',j,y}P_{u^n,{u^n}',j^n,y^n}(u,u',j,y)\log \frac{P_{y^n|u^n,{u^n}',j^n}(y|u,u',j)}{Q_{Y|UJ}(y|u,j)}\\
    &=\sum_{u,u',j,y}P_{u^n,{u^n}',j^n,y^n}(u,u',j,y)\log \frac{P_{y^n|u^n,{u^n}',j^n}(y|u,u',j)P_{u^n,{u^n}',j^n}(u,u',j)}{Q_{Y|UJ}(y|u,j)P_{u^n,{u^n}',j^n}(u,u',j)}.
\end{align}
Projecting the above equality onto $\mathcal{U}\times\mathcal{U}'\times\mathcal{Y}$ and by Pinsker's inequality, it follows that
\begin{align}
    \sum_{u,u',y}|\sum_{j}P_{y^n|u^n,{u^n}',j^n}(y|u,u',j)P_{u^n,{u^n}',j^n}(u,u',j) - \sum_{j}Q_{Y|UJ}(y|u,j)P_{u^n,{u^n}',j^n}(u,u',j)| \leq c\sqrt{3\eta}
\end{align}
Applying the same argument to $D(P_{y^n|u^{n'},j^{n'}}||Q_{Y|UJ})$ and $I(Y,S';U|J',U')$, we have
\begin{align}
    \sum_{u,u',y}|\sum_{j'}P_{y^n|u^n,{u^n}',{j^n}'}(y|u,u',j)P_{u^n,{u^n}',j^n}(u,u',j) - \sum_{j'}Q_{Y|UJ}(y|u',j')P_{u^n,{u^n}',{j^n}'}(u,u',j')| \leq c\sqrt{3\eta}
\end{align}
Using the triangle inequality, 
\begin{align}
    \sum_{u,u',y}P_{UU'}(u,u')|\sum_{j} Q_{Y|UJ}(y|u,j)P_{J|UU'}(j|u,u') - \sum_{j'}Q_{Y|UJ}(y|u',j')P_{J'|UU'}(j'|u,u')|\leq 2c\sqrt{3\eta}.
\end{align}
Now by the definition of good codewords and $\widetilde{R}$, any two good codewords $U^n$ and $U^{n'}$ satisfy
\begin{align}
    I(U^n;U^{n'}) \leq \widetilde{R} \leq D(Q_U) - \tau = \min_{\substack{Q_{UU'} s.t.:\\Q_U=Q_{U'},Pr\{U\overset{\mathcal{Q}}{\sim} U'=1\}}} I(U;U') - \tau.
\end{align}
This implies $Pr\{U\overset{\mathcal{Q}}{\sim} U'\} < 1$ and there exists some pairs $(u,u')$ that are not connected. For those pairs, by  \eqref{def: maximal error nonsymmetrizability},
\begin{align}
    \max_{y}|\sum_{j} Q_{Y|UJ}(y|u,j)P_{J|UU'}(j|u,u') - \sum_{j'}Q_{Y|UJ}(y|u',j')P_{J'|UU'}(j'|u,u')| > \varepsilon
\end{align}
for some positive $\varepsilon$. With sufficiently large $n$, one can choose a sufficiently small $\eta$ such that $\varepsilon > 2c\sqrt{4\eta}$, which leads to a contradiction.
This completes the proof of the unambiguity of the decoding rule.

\subsection{Error Analysis}
In the following we denote the selected codeword by $U^n$.
Define the following error events:
\begin{align}
    &\mathcal{E}_0 = \{P_{U^n(m,l),s^n,j^n}\notin \Psi^2_{\eta} \text{for all $l\in \mathcal{C}(m)$}\},\\
    &\mathcal{E}_1 = \{U^n \;\text{is not a good codeword}\},\\
    &\mathcal{E}_2 = \{P_{u^n,j^n,y^n}\notin \Psi^1_{\eta}\},\\
    &\mathcal{E}_3 = \{{U^n}' \;\text{is output by the decoder for some ${U^n}'\neq U^n.$}\}
\end{align}
\subsubsection{Error event $\mathcal{E}_0$}
By Lemma \ref{lem: joint typical lemma}, we have $Pr\{\mathcal{E}_0\}\to 0$ double exponentially fast with a sufficiently large $n$. 
\subsubsection{Error event $\mathcal{E}_1$}
The second error event is the case in which the encoder finds a codeword that is not good. By Lemma \ref{lem: good codeword} and Lemma \ref{lem: number of bad codewords in each subcodebook}, the number of bad codewords is upper bounded by $2^{-n\delta}2^{n\widetilde{R}}$. By a similar argument as \eqref{neq: probability within the bin} we have
\begin{align}
    Pr\{\text{$U^n$ is not a good codeword}\} \leq 2^{-n\nu}
\end{align}
for some positive $\nu>0$ and sufficiently large $n.$
\subsubsection{Error event $\mathcal{E}_2$} The probability $Pr\{\mathcal{E}_2|\mathcal{E}_0^c,\mathcal{E}_1^c\}\to 0$ directly follows from \eqref{neq: pujl qujl}.
\subsubsection{Error event $\mathcal{E}_3$}

Now suppose a codeword $u^n$ is selected with the jamming sequence $j^n$ and an error codeword ${u^n}'$ such that the joint type is $P_{UJU'}$. By Lemma \ref{lem: good codebook}, the number of such codeword ${u^n}'$ in the codebook is upper bounded by
\begin{align}\label{ine: number of wrong codewords maximal error}
    2^{n(|\widetilde{R}-I(U';U,J)|^++\delta)}.
\end{align}
The number of $(s^n,y^n)$ with joint type $P_{UU'SJY}$ such that
\begin{align}
    &D(P_{US}||Q_{US})\leq \eta,\\
    &D(P_{Y|USJ}||Q_{Y|USJ})\leq \eta,\\
    &\sum_{s,y}P_{UU'SJY} =P_{UU'J},\\
    \label{ine: maximal error mutual information lower bound for wrong codeword}&I(Y,S;U'|U,J) \geq\eta
\end{align}
is upper bounded by
\begin{align}\label{ine: number of yn sn maximal error}
    2^{nH(Y,S|U,U',J)}=2^{n(H(Y,S|U,J)-I(Y,S;U'|U,J))}.
\end{align}
Define the set of the corresponding conditional types $P_{YS|UU'J}$ by $\mathcal{D}_{YS|UU'J}$.
The error probability is upper bounded by
\begin{align*}
    &\sum_{P_{U'UJ}}\sum_{{u^n}'\in T_{P_{U'UJ}}^n[u^n,j^n]}\sum_{P_{YS|UU'J}\in \mathcal{D}_{YS|UU'J}}\sum_{(s^n,y^n)\in P_{YSUU'J}^n[u^n,{u^n}',j^n]} Pr\{S^n=s^n|u^n,j^n\} \sum_{x^n}W^n_{Y|XSJ}(y^n|x^n,s^n,j^n)Q^n_{X|US}(x^n|u^n,s^n)\\
    &\overset{(a)}{\leq}(n+1)^{|\mathcal{U}|^2|\mathcal{J}||\mathcal{S}||\mathcal{Y}|}2^{n(|\widetilde{R}-I(U';U,J)|^+ + \delta)}2^{n(H(Y,S|U,J)-I(Y,S;U'|U,J))}2^{-n(H(S|U, J)-\delta)}2^{-n(D(P_{Y|USJ}||Q_{Y|USJ})+H(Y|U,S,J))}\\
    &\overset{(b)}{\leq} 2^{-n(H(Y|U,S,J)+H(S|U,J)-H(Y,S|U,J)+I(Y,S;U'|U,J)-|\widetilde{R}-I(U';U,J)|^+ - \hat{\delta}(\delta))}\\
    &=2^{-n(I(Y,S;U'|U,J)-|\widetilde{R}-I(U';U,J)|^+ - \hat{\delta}(\delta))}
\end{align*}
where $(a)$ follows by Lemma \ref{lem: type counting lemma} (type counting lemma), \eqref{ine: number of wrong codewords maximal error}, \eqref{ine: number of yn sn maximal error} and the fact that for given $(u^n,j^n)$, the probability of any $s^n\in \mathcal{T}^n_{P_{SUJ}}[u^n,j^n]$ is the same and is upper bounded by $1/|\mathcal{T}^n_{P_{SUJ}}[u^n,j^n]|$, $(b)$ follows by the fact that the KL divergence is non-negative. When $R<I(U';U,J)$, by \eqref{ine: maximal error mutual information lower bound for wrong codeword} the error is bounded by
\begin{align}
 (n+1)^{|\mathcal{U}|^2|\mathcal{J}||\mathcal{S}||\mathcal{Y}|}2^{-n(\eta-\hat{\delta}(\delta))}.
\end{align}
When $\widetilde{R} > I(U';U,J)$, it follows that
\begin{align}
    &I(Y,S;U'|U,J)-\widetilde{R}+I(U';U,J) - \hat{\delta}(\delta)\\
    &=I(Y,S,U,J;U')-\hat{\delta}(\delta)-\widetilde{R}\\
    &\geq I(U';Y)-\hat{\delta}(\delta)-\widetilde{R}.
\end{align}
By the definition of the decoding rule there exists an ${s^n}'$ and a ${j^n}'$ such that 
\begin{align}
    &D(P_{y^n|{u^n}',{j^n}'}||Q_{Y|UJ}) \leq \eta,\\
    &D(P_{{u^n}',{s^n}',{j^n}'}||Q_{US}P_{J|US})\leq \eta\\
    &D(P_{{s^n}'}||Q_S)\leq \eta.
\end{align}
One can find a ${Y}'$ such that $P_{Y'U'J'}=\sum_{s}P_{Y'U'S'J'}\overset{\cdot}{=}P_{YU'S'J'}$
and we have $Pr\{\mathcal{E}_3|\mathcal{E}_0^c,\mathcal{E}_1^c,\mathcal{E}_2^c\} \leq 2^{-n\tau} $ for some $\tau >0$ if 
\begin{align}
    \widetilde{R} < I(U';Y')-\hat{\delta}(\delta).
\end{align}
 Considering the worst case, we have shown the achievability of 
\begin{align}
    \min_{P_{J|U}\in\mathcal{P}(\mathcal{J})}\min_{P_S\in\mathcal{P}(\mathcal{S},Q_S,\delta_n)} I(U;Y),
\end{align}
where 
\begin{align}
    \mathcal{P}(\mathcal{S},Q_S,\delta_n)=\{P_S\in\mathcal{P}(\mathcal{S}):D(P_S||Q_S)\leq \delta_n\}
\end{align}
and $\delta_n\to 0$ as $n\to \infty$.
With sufficiently large block length, we have $\delta_n \to 0$, and the proof is completed by combining the bound on $\widetilde{R}-R$. The distortion constraint is ensured by \eqref{ine: type bound of psjuy} and the typical average lemma.

\section{proof of Corollary \ref{coro: strictly causal maximal error}}\label{sec: proof of corollary strictly causal maximal error}

The decoding of $X^n$ only needs a normal maximal error code for the AVC $\{Q_{Y|XJ}:J\in\mathcal{J}\}$. Here we discuss the decoding of $U^n$. As we mentioned in the comments after Corollary \ref{coro: strictly causal maximal error}, the decoding of $U^n$ uses an average error code and is almost the same as the proof of Theorem \ref{the: achievability of strictly causal}. The difference is that we no longer have \eqref{ine: upper bound of i(usx;j)}. We can still prove that it is sufficient to consider $u^n$ and $s^n$ such that $I(U;J|X)< \epsilon$ and $I(S;J|X,U)< \epsilon$. However, for the maximal error we should not only consider $x^n$ such that $I(X;J)<\epsilon$, as those $x^n$ such that $I(X;S)\geq \epsilon$ almost do not contribute to the error probability and can be ignored only when we take the average over all $X^n$.

To solve this problem, we modify the definitions of $\Psi^3_{\eta}$ and $\Psi^4_{\eta}$ as
\begin{align}
    &\Psi^3_{\eta}=\{P_{YXUJ}: D(P_{YXJU} || Q_{YU|XJ}P_{XJ})\leq \eta\},\\
&\Psi^4_{\eta}=\{P_{XSJU}: D(P_{XSJU} || Q_{XUS}P_{J|X})\leq \eta\}.
\end{align}
To show the unambiguity of the decoding rule, we have
\begin{align}
    \eta \geq D(P_{YXUJ}||Q_{YU|XJ}P_{XJ}) = \sum_{x,u,j,u}P_{YXUJ}(y,x,u,j)\log \frac{P_{YXUJ}(y,x,u,j)}{Q_{YU|XJ}(y,u|x,j)P_{XJ}(x,j)}
\end{align}
and
\begin{align}
    \eta \geq I(U,Y;U'|J,X) = \sum_{u,u',x,j,y}P_{XUU'JY}(x,u,u',j,y)\log\frac{P_{U'|UXJY}(u'|u,x,j,y)}{P_{U'|XJ}(u'|x,j)}
\end{align}
Taking the sum of the above two inequalities gives
\begin{align}
    2\eta \geq \sum_{u,u',x,j,y}P_{XUU'JY}(x,u,u',j,y)\log\frac{P_{U'UXJY}(u',u,x,j,y)}{Q_{YU|XJ}(y,u|x,j)P_{U'|XJ}(u'|x,j)P_{XJ}(x,j)}.
\end{align}
Projecting it onto $\mathcal{U}\times\mathcal{U}'\times\mathcal{X}\times\mathcal{Y}$ and applying Pinsker's inequality,
\begin{align}
    &a\sqrt{2\eta}\\
    &\geq\sum_{u,u',x,y} |P_{U'UXY}(u',u,x,y) - \sum_j Q_{YU|XJ}(y,u|x,j)P_{U'|XJ}(u'|x,j)P_{XJ}(x,j)|\\
    &=\sum_{u,u',x,y} |P_{U'UXY}(u',u,x,y) - \sum_j \sum_s W_{Y|XSJ}(y|x,s,j)Q_S(s)Q_{U|XS}(u|x,s)P_{U'|XJ}(u'|x,j)P_{XJ}(x,j)|\\
    &\overset{(a)}{=}\sum_{u,u',x,y} |P_{U'UXY}(u',u,x,y) \\
    &\quad\quad\quad\quad - \sum_j \sum_s W_{Y|XSJ}(y|x,s,j)Q_S(s)\frac{Q_{XU}(x,u)Q_{S|XU}(s|x,u)}{Q_{X}(x)Q_{S}(s)}\frac{P_{J|U'X}(j|u',x)P_{U'X}(u',x)}{P_{XJ}(x,j)}P_{XJ}(x,j)|\\
    &=\sum_{u,u',x,y} |P_{U'UXY}(u',u,x,y) - Q_X(x)P_{U'|X}(u'|x)P_{U|X}(u|x)\sum_{j,s} W_{Y|XSJ}(y|x,s,j)Q_{S|XU}(s|x,u)P_{J|U'X}(j|u',x)|\\
    \label{eq: strictly causal maximal error unam 1}&=\sum_{u,u',x,y} |P_{U'UXJY}(u',u,x,j,y) - Q_X(x)P_{U'|X}(u'|x)P_{U|X}(u|x)\sum_{j} Q_{Y|XUJ}(y|x,u,j)P_{J|U'X}(j|u',x)|,
\end{align}
where $(a)$ follows by applying the Bayes rule to $Q_{U|XS}$ and $P_{U'|XJ},$ the last equality follows by defining $Q_{Y|XUJ}=\sum_{s} W_{Y|XSJ}Q_{S|XU}$. Similarly, for the wrong codeword $U^{n'}$ we have
\begin{align}
    &a\sqrt{2\eta}\\
    \label{eq: strictly causal maximal error unam 2}&\geq \sum_{u,u',x,y} |P_{U'UXY}(u',u,x,y) - Q_X(x)P_{U'|X}(u'|x)P_{U|X}(u|x)\sum_{j} Q_{Y|XUJ}(y|x,u',j')P_{J'|UX}(j'|u,x)|.
\end{align}
Applying the triangle inequality to \eqref{eq: strictly causal maximal error unam 1} and \eqref{eq: strictly causal maximal error unam 2} yields
\begin{align*}
    &2a\sqrt{2\eta}\geq\\
    &\sum_{x,u,u',y}Q_X(x)P_{U'|X}(u'|x)P_{U|X}(u|x)|\sum_{j} Q_{Y|XUJ}(y|x,u,j)P_{J|U'X}(j|u',x) - \sum_{j} Q_{Y|XUJ}(y|x,u',j')P_{J'|UX}(j'|u,x)|.
\end{align*}
By the assumption that $\min_x Q_X(x)\geq \beta_1$ and $\min_{x,u} Q_{U|X}(u|x)\geq \beta_2,$ we have
\begin{align}
    &\frac{2a\sqrt{2\eta}}{\beta_1\beta_2^2}\geq\sum_{x,u,u',y}|\sum_{j} Q_{Y|XUJ}(y|x,u,j)P_{J|U'X}(j|u',x)| - \sum_{j} Q_{Y|XUJ}(y|x,u',j')P_{J'|UX}(j'|u,x)|,
\end{align}
which is the same as \eqref{eq:  channel qyuj upper bound}. Hence, given the channel $Q_{Y|XUJ}$ is nonsymmetrizable-$\mathcal{U}|\mathcal{X}$, the remaining proof is the same as that in Appendix \ref{sec: proof of achievability of strictly causal}. The only difference is that when considering the worst case regarding the jammer, the jamming strategy is $Q_{J|X}$ instead of $Q_J$ as we replace $Q_XP_J$ with $P_{JX}=Q_X P_{J|X}$ in the definitions of $\Psi_{\eta}^3$ and $\Psi_{\eta}^4$.
\section{proof of corollary \ref{coro: noncausal binary output lower bound}}

Given the channel $\mathcal{W}=\{W_{Y|XSJ}\}$ indexed by $j\in\mathcal{J}$ and $Q_S$ with input distributions $Q_{U|S},Q_{X|US}$, we define a new arbitrarily varying channel $\mathcal{Q}=\{Q_{Y|UJ}\}$ such that $Q_{Y|UJ}=\sum_{x,s}W_{Y|XSJ}Q_{X|US}Q_{S|U}$ and restrict the input alphabet $|\mathcal{U}|=2$. Consider the row-convex extension of $\mathcal{Q}$ defined by
\begin{align}
    \overline{\overline{\mathcal{Q}}}=\{Q_{Y|U}\},
\end{align}
where
\begin{align}
    Q_{Y|U}(y|u) = \sum_{j} Q_{Y|UJ}(y|u,j)Q_{J|U}(j|u)
\end{align}
where $Q_{J|U}$ can be any stochastic matrix. 

By Corollary 12.3 and Theorem 12.6 in \cite{csiszar2011information}, it follows that
\begin{align}
    C_m(\mathcal{Q}) = C_m(\overline{\overline{\mathcal{Q}}}) = \min_{Q\in\overline{\overline{Q}}}C(Q),
\end{align}
where $Q\in \overline{\overline{\mathcal{Q}}}$ are normal discrete memoryless channels.
Furthermore, there exists a DMC $Q^*\in \overline{\overline{\mathcal{Q}}}$ such that for any code with the standard minimum distance (SMD) decoder, the error probability of applying this code to any channel in $\overline{\overline{\mathcal{Q}}}^n$ is upper bounded by those over $Q^*$\cite[Theorem 12.6]{csiszar2011information}. Suppose there exists a random constant-composition code $(n,f,\phi)$ with an SMD decoder $\phi$ such that
\begin{align}
    e({Q^*}^n,f,\phi) \leq 2^{-n\nu}.
\end{align}
We have
\begin{align}
    e(\overline{\overline{\mathcal{Q}}}^n,f,\phi)\leq e({Q^*}^n,f,\phi) \leq  2^{-n\nu}.
\end{align}
On account of Lemma 12.3\cite{csiszar2011information}, we also have
\begin{align}
    e(\mathcal{Q}^n,f,\phi) = e(\overline{\overline{\mathcal{Q}}}^n,f,\phi)\leq 2^{-n\nu}.
\end{align}
Hence, it is also a maximum error code for the arbitrarily varying channel $\mathcal{Q}$. Note that as implied by \cite[Lemma 12.5]{csiszar2011information}, the SMD decoder for the channel $Q^*$ is equivalent to a maximum likelihood decoder, and we can always use a random coding argument to construct such a code. We show that one can construct a Gel'fand-Pinsker code for the corresponding AVC $\mathcal{W}$ based on the code for $\mathcal{Q}$. To this end, we partition the codebook into $2^{nR}$ sub-codebooks $\mathcal{C}(m)$ such that the size of each subcodebook satisfies $\frac{1}{n}\log|\mathcal{C}(m)| > I(U;S)$. Each codeword $u^n$ now can be indexed by a unique pair $(m,l)$, where $l$ is the index of $u^n$ within the subcodebook $\mathcal{C}(m).$ We use a sub-optimal decoding strategy that requires the decoder to estimate both $m$ and $l$. For each pair of $(u^n,s^n)$, the encoder generates the channel input $x^n$ by the product law $Q^n_{X|US}=\prod_{i=1}^n Q_{X|US}$. The error probability of a given $u^n$ averaged over all 
random state sequences $s^n$ is 
\begin{align}
    \sum_{x^n,s^n}W^n((\phi^{-1}(m,l))^c|x^n,s^n,j^n)Q^n_{X|US}(x^n|u^n,s^n)Pr\{s^n|u^n\}
\end{align}
The probability $Pr\{s^n|u^n\}$ satisfies
\begin{align}
    Pr\{s^n|u^n\} \leq 2^{-n(H(S|U)-\epsilon)}
\end{align}
for $s^n\in\mathcal{T}^n_{Q_{S|U},\delta}[u^n]$ and $0$ otherwise. It follows that
\begin{align}
    Pr\{s^n|u^n\} \leq \prod_{i=1}^n Q_{S|U}(s_i|u_i) 2^{n\hat{\epsilon}}
\end{align}
for some $\hat{\epsilon}>0.$
To prove this, for $s^n\notin \mathcal{T}^n_{Q_{S|U}}[u^n]$ the inequality always holds. When $s^n\in \mathcal{T}^n_{Q_{S|U},\delta}[u^n]$, we have
\begin{align}
    \prod_{i=1}^n Q_{S|U}(s_i|u_i) \leq 2^{-n(H(S|U) - \epsilon(\delta))}.
\end{align}
By setting $\hat{\epsilon}= [\epsilon(\delta)-\epsilon]^+$, the proof is completed.

Now, we have
\begin{align}
    &\sum_{x^n,s^n}W^n((\phi^{-1}(m,l))^c|x^n,s^n,j^n)Q^n_{X|US}(x^n|u^n,s^n)Pr\{s^n|u^n\} \\
    &\leq \sum_{x^n,s^n}\sum_{y^n\in(\phi^{-1}(m,l))^c } \prod_{i=1}^n W(y_i|x_i,s_i,j_i) Q_{X|SU}(x_i|s_i,u_i) Q_{S|U}(s_i|u_i)2^{n\hat{\epsilon}}\\
    &=\sum_{y^n\in(\phi^{-1}(m,l))^c } \prod_{i=1}^n Q_{Y|UJ}(y_i|u_i,j_i)2^{n\hat{\epsilon}}\\
    &=Q^n_{Y|UJ}((\phi^{-1}(m,l))^c|u^n(m,l),j^n)2^{n\hat{\epsilon}}\\
    &\leq 2^{-n(\nu-[\epsilon(\delta)-\epsilon]^+)}.
\end{align}
The result can be extended to a larger $|\mathcal{U}|$ by \cite[Theorem 12.7]{csiszar2011information}, and this completes the proof.

\bibliographystyle{ieeetr} 
\bibliography{ref}
\end{document}